\documentclass[12pt,a4paper,reqno]{amsart}

\usepackage{amssymb,amsthm}
\usepackage{amsmath}
\usepackage{hyperref}

\def\lrp#1{\left( #1\right)}

\def\bm#1#2#3#4{\begin{pmatrix}#1&#2\\#3&#4\end{pmatrix}}
\def\lra#1{\left\langle #1\right\rangle}
\def\lrae#1#2{{\lra{#1,\, #2 #1}}}

\begin{document}

\theoremstyle{plain}
\newtheorem{proposition}{Proposition}[section]
\newtheorem{theorem}[proposition]{Theorem}
\newtheorem{lemma}[proposition]{Lemma}

\theoremstyle{definition} 
\newtheorem{definition}[proposition]{Definition}

\theoremstyle{remark}
\newtheorem{remark}[proposition]{Remark}

\renewcommand{\thesection}{\Roman{section}}
\renewcommand{\thesubsection}{\thesection.\arabic{subsection}}
\renewcommand{\theequation}{\thesection.\arabic{equation}}

\title[Complex Classical Fields]{\bf  Complex  Classical Fields:\\
A Framework for Reflection Positivity}

\author[A. Jaffe]{Arthur Jaffe}
\address{Harvard University\\
Cambridge, MA 02138}
\email{arthur\_jaffe@harvard.edu}

\author[C. D. J\"{a}kel]{Christian D.\ J\"{a}kel}
\address{School of Mathematics\\ 
Cardiff University, Wales}
\email{christian.jaekel@mac.com} 

\author[R. E. Martinez II]{Roberto E. Martinez II}
\address{Harvard University\\
Cambridge, MA 02138}
\email{remartin@fas.harvard.edu}

\begin{abstract}
We explore a framework for complex classical fields, appropriate for describing quantum field theories.  Our fields are linear transformations on a Hilbert space, so they are more general than random variables for a probability measure. Our method generalizes  Osterwalder and Schrader's construction of Euclidean fields.  We allow complex-valued classical fields in the case of quantum field theories that describe neutral particles.  

From an analytic point-of-view, the key to using our method is {\em reflection positivity}. We investigate conditions on the Fourier representation of the fields to ensure that reflection positivity holds. We also show how reflection positivity is preserved by  various spacetime compactifications of~$\mathbb{R}^{d}$ in different coordinate directions.
\end{abstract}

\maketitle

\tableofcontents

\section{Classical Fields}
\subsection{Overview}
We study a Fock-Hilbert space $\mathcal{E}$ on which averaged classical fields act as linear transformations.  These fields generate an abelian algebra of unbounded operators that are defined on a common, dense, invariant domain.  We call a classical field \emph{neutral}/\emph{charged} if it arises in the description of neutral/charged particles. What is special in our framework is that our neutral fields can be either real or \emph{complex}, so the usual distinction between real and complex fields does not coincide here with the distinction between neutral and charged fields.  Real neutral fields reduce to the usual case;  complex neutral fields allow for something new.\footnote{\textrm{ Parts of this work were carried out during authors' visits to  the  Institute for Theoretical Physics, ETH Z\"urich,  the Erwin Schr\"odingier Institute (ESI), Vienna,  and the mathematics and physics departments at Harvard University.  The relevant authors thank these institutions both for hospitality and for 
support.}}  

The two-point function of fields is the integral kernel of an operator~$D$, which need not be hermitian.  We require two things: First, the hermitian part of $D$ should have strictly positive spectrum. Second, the transformation $D$ should be {reflection positive}. As precise
formulation of reflection positivity is provided in Definition~\ref{Def:Neutral-OS Positivity} (for the neutral case) and Definition~\ref{Defn:Charged RP} (for the charged case).

We also deal with charge in a somewhat novel way.  In the usual case, a charged field is represented as a complex-linear combination of two  neutral fields. The usual charge conjugation arises as the complex conjugation of this field; this can be implemented as a unitary operator.  However, in this paper we introduce distinct charged fields~$\Phi_{\pm}$---{not related by complex conjugation}---whose labels correspond to the fundamental charge carried by the field.  Charge conjugation acts as a unitary transformation $\mathbf{U}_{c} $ on~$\mathcal{E}$ 
such that  $\mathbf{U}_{c} \, \Phi_{\pm} \mathbf{U}_{c}^{-1} =\Phi_{\mp} \, $.  

\subsection{More Details}
Kurt Symanzik introduced the concept of a Eu\-clidean-invariant Markoff field associated with an underlying probability distribution of classical fields~\cite{Symanzik}.  Euclidean-covariant classical fields describing neutral scalar particles are typically real, as are time-zero quantum fields.  With the standard distributions that occur for scalar quantum fields, the zero-particle expectations of products of fields are typically positive.  This is also that case for the vacuum expectation values of time-ordered products of imaginary-time quantum fields (\emph{i.e.}, the analytic continuation of quantum fields in Minkowski space). In fact, the analytic continuation of anti-time-ordered vacuum expectation values of quantum fields in Minkowski space should agree with the expectations of such classical fields.  

Edward Nelson formulated a set of mathematical axioms interpreting Euclidean Markoff fields as random variables.  A Markoff field satisfying these axioms yields  a corresponding quantum field~\cite{Nelson Reconstruction}.  Although these axioms apply beautifully to the free scalar field \cite{Nelson-Free_Field}, and to some other cases, verifying the global Markoff property for known examples of interacting scalar fields poses certain difficulties.  Moreover, an analogous set of axioms has not been formulated for  fermionic or gauge fields. 

Konrad Osterwalder and Robert Schrader discovered an alternative and more-widely applicable  approach based on a property that they called  {\em reflection positivity} (RP),  which today is often called Oster\-walder-Schrader (OS) positivity.   Every Euclidean-invariant, OS-posi\-tive and regular set of expectations  yields a relativistic, local quantum field theory~\cite{Osterwalder-Schrader, Osterwalder-Schrader II}. Assuming  certain growth conditions in both the quantum and classical framework, the OS axioms for a classical field theory were shown to be equivalent to the Wightman axioms for a corresponding quantum field. According to Zinoviev, these growth assumptions can be replaced by a {\em weak spectral condition} \cite{Zinoviev}.

While RP allows one to give a rigorous meaning to {\em  inverse Wick rotation}, RP has also found applications in various areas of mathematics outside of mathematical physics.  For example, RP has had an impact on the theory of analytic continuation of group representations; 
recent results and extensive references can be found in \cite{JO}. Moreover, within mathematical physics, RP has had an enormous impact in statistical physics---especially in understanding  properties of the spectrum of the transfer matrix, as well as in the theory of phase transitions. A contemporary review can be found in \cite{B}. 

\subsection{This Work}

In this work we demonstrate that by no longer insisting  that neutral fields are real,  one gains a great deal of added flexibility. A  
number of relevant problems, which so far were not accessible from Euclidean quantum field theory, can now be formulated  in terms of classical
fields. The new framework allows us to consider Hamiltonians with complex-valued heat kernels such as, for example, $H=H_{0}+\vec v\cdot \vec P$, which for  $| \vec v |<1$ equals (up to an overall multiplicative constant) 
the Hamiltonian $H_{0}$ (in a finite spatial volume) as seen from a Lorentz frame moving with velocity~$\vec v$.  
In fact, this was our motivating example, and it arose in our attempt to understand the work of Heifets and Osipov~\cite{Osipov}.  We consider this example in detail in a separate publication  \cite{JJM-PartII}.   

We also study charged fields.  We introduce fields $\Phi_{\pm}$ that differ from the usual ones in that charge conjugation $\Phi_{\pm}\to\Phi_{\mp}$ is given by a unitary transformation {\em different} from complex conjugation. We apply this framework to study the 
thermal equilibrium states of a charged field with a chemical potential at positive temperature in the forthcoming article \cite{JJM-PartII}. 
Such states and fields occur in the study of the statistical mechanics of Bose-Einstein condensation; see \cite{BFKT}.   

The expectation values of our classical fields are defined on a ``Euclidean Hilbert space.''  In this paper we consider  Euclidean Fock space, 
which is the simplest case. The classical fields are represented as unbounded operators and finiteness of their expectation values follows once 
certain operator domain questions have been resolved. In case the usual description in terms of functional 
integrals is available, the domain questions  in our approach {\em can be} resolved, and the two descriptions are equivalent.  Our method is similar to the construction of Euclidean  fields given by Osterwalder and Schrader  \cite{O-S-Fields}.  We  require~RP for the no-particle expectation. 

We begin in \S\ref{Sect:NeutralCoordinates} by studying what transformation properties of the field $\Phi$ under reflections are equivalent to RP.   In \S\ref{Sect:Quantization} we briefly describe the associated quantization in the Gaussian case.  We generalize this for charged fields in~\S\ref{Sect:Charged Fields}. In
\S\ref{Sect:Compactification_RP} we show that  RP on $\mathbb{R}^d$ gives RP on spacetimes~$\mathbf{X}$
that are compactified in one or more coordinate directions.

Certain non-Gaussian expectations can then be built as perturbations of the Gaussian case. These examples can 
be studied using a cut-off and a generalised Feynman-Kac formula \cite{JJM-PartII}.   We rely on RP to obtain a 
robust inner product that one can use to prove useful estimates. This inner product also provides  the relation between 
the Euclidean Fock space $\mathcal{E}$ and the Hilbert space~$\mathcal{H}$ of quantum theory.  

\subsection{The Problem with Measures}
One might expect  that the com\-plex-valued Schwinger functions are moments of a complex measure.  But even in the Gaussian examples we consider, such a countably-additive measure does not exist.  Without measure theory one loses the possibility to use $L_p$ estimates to study convergence of integrals, and thus one loses quantitative control of the theory.  

The standard formulation of a real classical field $\Phi\in\mathcal{S}'_{\rm real}(\mathbb{R}^d)$, or a complex field $\Phi\in\mathcal{S}'(\mathbb{R}^d)$, is to take it to be a random variable for a probability measure $d\mu(\Phi)$ on
the space of tempered distributions $\mathcal{S}'_{\rm real}(\mathbb{R}^d)$
or $\mathcal{S}'(\mathbb{R}^d)$.  In the real case the characteristic function of the measure $d\mu$ is given by
\begin{equation}
            S(f)
            = \int e^{i\Phi(f)}\,
                d\mu(\Phi)\;.
        \label{Gaussian_Characteristic_function}
\end{equation}
The exact criterion for the existence of a countably-additive probability measure $d\mu(\Phi)$  on a nuclear space (such as 
$\mathcal{S}_{\rm real}(\mathbb{R}^d)$ or  $\mathcal{S}(\mathbb{R}^d)$) is the following:

\begin{proposition}[\bf Minlos' Theorem \cite{Minlos}]
\label{Prof:Minlos_Theorem} A functional $S(f)$  on a nuclear space 
$\mathcal{S}$ 
is the characteristic
functional of a countably-additive, probability measure on the dual space $\mathcal{S}'$, if and only if $S(f)$ 
is continuous, of positive type, and normalized by
$S(0)=1$.  
\end{proposition}

Complex measures are more delicate
mathematically.  Borchers and Yngvason \cite{Borchers-Yngvason} studied possibilities for measures being associated with arbitrary Wightman field theories.  In  the Gaussian case, there
is a clean result for the existence of a complex Gaussian measures. Suppose $d\mu$ is a  complex-valued Gaussian with mean zero and covariance $D$, and the resulting characteristic function is given by
\begin{equation}
		S_{D}(f)
            	=e^{-\tfrac{1}{2}\, \langle \bar f,Df \rangle_{L_2}} \;.
\end{equation}
Then one has the existence criterion of Proposition 4.4 in \cite{Yngvason}:
    \begin{proposition}[\bf Yngvason's Criterion]
    \label{Prop:Yngvason}
    Let $0 < K=K^*$ and $L=L^*$ be  bounded transformations on
    $L_2(\mathbb{R}^d)$ and continuous transformations on $\mathcal{S}(\mathbb{R}^d)$.   Then there
    exists a countably-additive, complex-valued, Gaussian
    measure $d\mu_D(\Phi)$ on $\mathcal{S}'_{\rm real}(\mathbb{R}^d)$ with covariance 
    $D=K+iL$ and with characteristic function $S_{D}(f)$
    if and only if $K^{-1/2}LK^{-1/2}$ is Hilbert-Schmidt on $L_2(\mathbb{R}^d)$.
    \end{proposition}

This criterion is linked to the desire to write  $d\mu_{D}$ as a phase times a probability measure $d\mu_{G}$ normalized by the constant $\mathfrak{Z}$, namely 
\begin{equation}
        d\mu_D(\Phi) = \frac{1}{\mathfrak{Z}}\,e^{\frac{i}{2} \langle \Phi, Y\Phi \rangle }\,
            d\mu_G(\Phi)\;,
         \text{so  } 
                  |d\mu_D(\Phi)|
        \leqslant\frac{1}{ | \mathfrak{Z}| }\,
        d\mu_G(\Phi)\;.
    \label{Complex-Measure}
\end{equation}
To obtain insight into Yngvason's Criterion, let us assume for simplicity that $K$ and $L$ commute, and that the spectrum of $K^{-1/2}LK^{-1/2}=LK^{-1}$ is even.   Then one finds 
	\[
	D^{-1}= |D|^{-2}K
		-i |D|^{-2}L\;,
	\quad
	G^{-1}= |D|^{-2}K\;,
	\quad
	Y= |D|^{-2}L\;
	\]
and
	\[
	\mathfrak{Z}= \left( \det\lrp{I-iLK^{-1}} \right)^{1/2}\;.
	\]
Denote the positive eigenvalues of $LK^{-1}$
by $\lambda_j$.  Then  
\begin{equation}
        \mathfrak{Z}
        =\lrp{\prod_{j} (1+i\lambda_j)(1-i\lambda_j)}^{1/2}
        =\lrp{\prod_{j} (1+\lambda_j^2)}^{1/2}\geqslant 1\;.
    \label{ProductRepn}
\end{equation}
In \eqref{ProductRepn} the product defining $\mathfrak{Z}$ converges if and only if   $LK^{-1}$ is Hilbert-Schmidt, which is Yngvason's Criterion.

In the examples we consider in \cite{JJM-PartII}, Yngvason's Criterion does {\em not} apply.  In these examples not  only is the spectrum of $LK^{-1}$ continuous (and hence not Hilbert-Schmidt), but also an infrared cutoff would yield eigenvalues $\lambda_{j}$ that  would not converge to zero.  In this case, as well as in the continuous case without a cutoff, $\mathfrak{Z}$  is infinite.\footnote{Related issues arise in
Brydges and Imbrie's study of random walks; see Equation
(7.2) of \cite{Brydges-Imbrie}. We thank John Imbrie for
bringing this to our attention.}

\setcounter{equation}{0}
\section{Classical Fields as Operators on Hilbert Space
\label{Sect:NeutralFields}} One can define a neutral,
random field by introducing the Fock-Hilbert space
$\mathcal{E}=\mathcal{E}(\mathcal{K})$ over a one-particle space  $\mathcal{K}$. This exponential Hilbert space has the form
\begin{equation}
        \mathcal{E}
        = \bigoplus_{n=0}^\infty \,\mathcal{E}_n\;,
        \quad \text{where} \quad
        \mathcal{E}_{0}=\mathbb{C}\;,
        \quad \text{and} \quad
          \mathcal{E}_{n}
          =  \underbrace{\mathcal{K} \otimes_s
                \cdots\otimes_s \mathcal{K}}_{n\ \rm factors}\;,
    \label{Euclidean Fock Space}
\end{equation}
and where $\otimes_s$ denotes the symmetric tensor product. Let the distinguished vector $\Omega_{0}^{\tt E}=1\in\mathcal{E}_{0}$ denote the zero-particle state.

In the following, we take $\mathcal{K}=L_2(\mathbf{X})$, but one could just as well take $\mathcal{K}= \bigoplus_{j=1} ^{ N} L_{2}(\mathbf{X})$ for an $N$-component field.  Either
$\mathbf{X}$ denotes Euclidean spacetime $\mathbb{R}^d$, a toroidal spacetime $\mathbb{T}^d=S^{1}\times \cdots \times S^{1}$,
 or more generally an intermediate case in which spacetime has the form $\mathbf{X}=X_1\times \cdots X_d$, where each factor~$X_j$ 
either equals $\mathbb{R}$ or a circle $S^{1}$ of length $\ell_j$. We let $\mathcal{S}(\mathbf{X})$ denote the $C^{\infty}$ functions on $\mathbf{X}$ 
with the topology given by the usual family of seminorms $\|f \|_{r,s}=\sup_{x\in \mathbf{X}} | x^{r}D^{s}f(x)| $. 
Then $\Phi\in\mathcal{S}'(\mathbf{X})$, the dual space of continuous linear functionals on $\mathcal{S}(\mathbf{X})$, 
and it pairs linearly with test functions $f\in\mathcal{S}(\mathbf{X})$ yielding $\Phi(f)=\int \Phi(x) f(x)dx$. 

We assume the neutral random field  $\Phi$ is an operator-valued distribution on the Hilbert space~$\mathcal{E}$ with each $\Phi(f)$ for  $f\in\mathcal{S}(\mathbf{X})$ defined on a common, dense, invariant domain $\mathcal{D}\subset \mathcal{E}$.   Note that~$\Phi(f)^{*}=\Phi^{*}({\overline f})$ on the domain $\mathcal{D}$. 
 This domain includes $\Omega_{0}^{\tt E}$ and is invariant under the action of the field, namely $\Phi(f)\mathcal{D}\subset\mathcal{D}$. The characteristic functional of a neutral field is defined as the exponential
\begin{equation}
            S(f)
            = \lrae{\Omega_0^{\tt E}}{e^{i\Phi(f)}\,}_\mathcal{E}
            = \sum_{n=0}^\infty \frac{i^n}{n!}\,
                    \lrae{\Omega_0^{\tt E}}{\Phi(f)^n\,}_{\mathcal{E}}\;,
        \label{Characteristic_Functional on E}
\end{equation}
and we assume this series converges for all $f\in\mathcal{S}(\mathbf{X})$. 
In the Gaussian case with mean zero, $\langle \Omega_0^{\tt E},  \Phi(f)^{2n}  \Omega_0^{\tt E} \rangle_\mathcal{E}= (2n-1) !!\langle \Omega_0^{\tt E},  \Phi(f)^{2}  \Omega_0^{\tt E} \rangle^n$
demonstrating convergence explicitly. We reinterpret this for  charged fields in \S\ref{Sect:Charged Fields}.
    
For Gaussian fields with bounded two point function $D$, the exponential series for the vector $e^{i\Phi(f)}\Omega_0^{\tt E}$ converges strongly, and the characteristic function \eqref{Gaussian_Characteristic_function} is well-defined by its exponential series. In the case of real, neutral fields, non-Gaussian expectations, associated with an action $\mathfrak{A}$, approximated by cut-off actions $\mathfrak{A}_{n}$, have been constructed. 
One replaces the expectation $\lrae{\Omega_{0}^{\tt E}}{\ \cdot\ \, }_{\mathcal{E}}$ with a sequence of  normalized expectations of the form  
	\[ 
		\omega_{n} (\ \cdot\ )  = 
		 \frac{ \lrae{\Omega_{0}^{\tt E}}{\ \cdot\ e^{-\mathfrak{A}_{n}}\, }_{\mathcal{E}}}{
		 \lrae{\Omega_{0}^{\tt E}}{e^{-\mathfrak{A}_{n}}\, }_{\mathcal{E}}} \; ,  
	\]
yield  another translation invariant expectation in the limit $n\to\infty$. Some corresponding Euclidean Hilbert spaces have been given by {\em constructive quantum field theory}; see \cite{Glimm-Jaffe}.

\subsection{Neutral, Classical Fields: the Case  $ \mathbf{X}=\mathbb{R}^{d}$ \label{Sect:NeutralCoordinates}}
In this section, we consider fields on Euclidean spacetime, $\mathbf{X}=\mathbb{R}^{d}$, with the one-particle space $\mathcal{K}=L_{2}(\mathbb{R}^{d})$. The  {\em annihilation operator} (which is actually a densely-defined bilinear form on $\mathcal{E}\times\mathcal{E}$) has non-vanishing matrix elements from $\mathcal{E}_n$ to $\mathcal{E}_{n-1}$.  In the Fourier representation it acts as 
\[ (A(k)f)_{n-1}(k_1,\ldots,k_{n-1})=\sqrt{n}\,f_n(k,k_1,\ldots,k_{n-1}) \;.  \]
Then $[A(k), A(k')]=0 $. 
The adjoint creation form $A(k)^*$ satisfies the usual canonical
relations, $[A(k),A(k')^*]=\delta(k-k')$.   Define the complex coordinates
\[
        \widetilde Q (k)
        = A(k)^*+A(-k)\;.
\]
These coordinates mutually commute, $ [\widetilde Q (k),\widetilde Q (k')]=0$, and also  
\begin{equation}
    	    \widetilde Q^{*} (k)
	    = \widetilde Q (-k)\;,\quad \text{and}\quad
	    \lrae{\Omega_0^{\tt E}}{\widetilde Q (k)\,\widetilde Q (k')\,}_\mathcal{E}
	    =\delta(k+k')\;.
    \label{Q(-k)}    
\end{equation}
    
The Gaussian coordinate field $Q(x)$ is the Fourier transform of $\widetilde Q(k)$,  namely 
\begin{equation}
        Q(x)
        = \lrp{2\pi}^{-d/2}
            \int \widetilde Q (k)\,e^{ik\cdot
            x}\,dk\;.
    \label{Q-Field_Form}
\end{equation}
Take the general Gaussian, neutral, scalar, classical field $\Phi(x)$ to  be  a linear
function of $Q (x)$.  Assume that for some given function $\widetilde \sigma (k)$,
\begin{equation}
        \Phi(x)
         = \lrp{2\pi}^{-d/2}
            \int \widetilde Q (k)\, \widetilde{\sigma}(k) \,e^{ik\cdot
            x}\,dk\;.
    \label{Field_Form}
\end{equation}
Whatever the choice of $\widetilde \sigma(k)$, the fields $\Phi(x)$ and their adjoints mutually commute,
 \begin{equation}
        [ \Phi(x) , \Phi(x')]
        =  [\Phi(x) , \Phi^*(x')]=0\;.
    \label{Field_Commutator}
\end{equation}
The field $\Phi(x)=\Phi^{*} (x)$ is hermitian in case that
$\widetilde{\sigma}(k)=\overline{\widetilde{\sigma}(-k)}$,
or, equivalently, when $\sigma=\overline\sigma$ is
real\footnote{For a transformation $S$ on $L_2$ with
integral kernel $S(x;x')$, the integral kernel of the
transpose  $S^{\rm T}$ is $S(x';x)$, the kernel of the
complex-conjugate $\overline S$ is $\overline{S(x;x')}$,
and the kernel of the hermitian-adjoint $S^*$ is
$\overline{S(x';x)}$. The operator $S$ is defined to be
{\em symmetric} if $S=S^{\rm T}$, is defined to be {\em real}
if $S=\overline S$, and is defined to be {\em hermitian} if
$S=S^*$. The kernel of a translation-invariant operator has
the form $S(x;x')=S(x-x')$.  In Fourier space:
    \[
        \widetilde{S^{\rm T}}(k)=\widetilde S(-k)\;,\qquad
        \widetilde{\overline S}\,(k)=\overline {\widetilde S(-k)}\;,
        \qquad
        \text{and}
        \widetilde{ S^*}\,(k)=\overline{\widetilde S(k)}\;.
    \]}.    
In the standard free, Euclidean-field example, one takes $\widetilde\sigma(k)=(k^{2}+m^{2})^{-1/2}$.  
    
In configuration space, the definition \eqref{Field_Form}  amounts to the relation
\begin{equation}
        \Phi(x)={(2\pi)^{-d/2}}\lrp{\sigma\,Q}(x)\;.
    \label{Field_as_SQ}
\end{equation}
Here, one defines the convolution operator $\sigma$  by
\[
        (\sigma f)(x)
        = (2\pi)^{-d/2}\,\int\sigma(x-x')f(x')dx'
\]
where
\[
        \sigma(x)
        = (2\pi)^{-d/2}\,\int
        \widetilde\sigma (k) \,e^{ik\cdot x}dk\;.
\]
The expectation of two fields defines an
operator $D \colon \mathcal{K}\mapsto \mathcal{K}$ with integral kernel
\[
        D(x,x')
        = \lrae{\Omega_0^{\tt E}}{\Phi(x)\Phi(x')\,}_\mathcal{E}\;.
\]

One can introduce commuting, canonically-conjugate coordinates (which
generally do not enter the functional integrals), namely 
\[
        \widetilde P(k)
        = \tfrac{i}{2}\lrp{A(k)^*-A(-k)}
        = \widetilde P(-k)^*\;,
\]
so
\[
        [\widetilde P(k) , \widetilde Q(k') ]
        = -i\delta(k+k')\;.
\]
In case that $\sigma^{\rm T}$ is invertible,  the conjugate field is 
\[
        \Pi(x)
        = (2\pi)^{d/2} (\sigma^{{\rm T}})^{-1}P(x)
        = (2\pi)^{-d/2} \int \widetilde\sigma({-k})^{-1}\widetilde P(k)
            e^{ik\cdot x}\,dk\;.
\]
With these conventions,
$[\Pi(x) , \Phi(x')]=-i\delta(x-x')$.

\subsection{Schwinger Functions}
The zero-particle expectations 
\[
S_{n}(x_{1},\ldots,x_{n})=\lrae{\Omega_{0}^{\tt E}}{\Phi(x_{1})\cdots\Phi(x_{n})\,}_\mathcal{E}
\]
satisfy a Gaussian recursion relation, 
\begin{equation}
		S_{n}(x_{1},\ldots,x_{n})=\sum_{j=2}^{n} S_{2}(x_{1},x_{j})\,
			S_{n-2}(x_{2}, \ldots, \not\hskip-2pt x_{j}, \ldots, x_{n})\;,
	\label{S-Recursion}
\end{equation}
where $\not\hskip-2pt x_{j}$ denotes the omission of $x_{j}$ and 
\begin{equation}       
	S_{2}(x,x')
       = D(x-x')
        = (2\pi)^{-d}\,
            \int \widetilde{\sigma}(k)\,\widetilde{\sigma}(-k)\,e^{ik\cdot(x-x')}\,dk  \;.
    \label{TwoPointIdentification}
\end{equation}
Commutativity of the fields assures that $D(x)$ is an even function,   $D(x-x')=D(x'-x)$.  So
\begin{equation}
        D=\sigma\sigma^{\rm T}
        = \sigma^{\rm T}\sigma
        = D^{\rm T}\;,
       \quad \text{and}\quad       
        (2\pi)^{d/2}\widetilde D(k)
        = \widetilde{\sigma}(k)\,\widetilde{\sigma}(-k)\;.
    \label{Neutral Covariance}
\end{equation}
Given  $\widetilde D(k)$, the general solution for $\widetilde{\sigma}(k)$ is 
\[
        \widetilde{\sigma}(k)
        = (2\pi)^{d/4} \widetilde D(k)^{1/2}\,e^{h(k)}\;,
\]
with $h(k)$ an odd function $h(k)=-h(-k)$.  Here, we take the straightforward choice $h(k)=0$, so
	\[ 
       \widetilde\sigma(k)
       = (2\pi)^{d/4} \widetilde D(k)^{1/2}
       =\widetilde \sigma(-k)\;,
	\]
and thus \eqref{Q-Field_Form} equals
\begin{equation}
        \Phi(x)
        = \lrp{2\pi}^{-d/4}
            \int \widetilde Q(k)\, \widetilde D(k)^{1/2}
             \,e^{ik\cdot x}\,dk\;.
   \label{Classical_Field-1}
\end{equation}
Note that the recursion relation \eqref{S-Recursion} ensures that 
\[
		\langle \Omega_{0}^{\tt E}, \Phi(f)^{2n} \Omega_{0}^{\tt E} \rangle_{\mathcal{E}}
		= (2n-1)!! \,\langle \overline f, Df \rangle^{n}_{L_{2}} \;.
\]

We want to ensure that multiplication by $\widetilde D(k)^{1/2}$ defines a
continuous transformation of Schwartz space~$\mathcal{S}(X)$ into itself.  This requires choosing  an
appropriate square root. Consider the case  $\widetilde
D(k)= \widetilde K(k)+i\widetilde L(k)$, with $0<\widetilde
K$, and $\widetilde L$ real. One can use the
positive square root of $\widetilde K(k)$ to write
    \[
       (2\pi)^{-d/4} \widetilde \sigma(k)
        = \widetilde D(k)^{1/2}
        =\widetilde K(k)^{1/2}\lrp{1+i\, \widetilde L(k)\widetilde K(k)^{-1} }^{1/2}\;.
       \]
Here, one chooses the square root $\widetilde D(k)^{1/2}$ so
that $\widetilde K(k)^{1/2}$ is positive, and the real part of the term
$(1+i\, \widetilde L(k)\widetilde K(k)^{-1} )$   lies in the complex half-plane with real part
greater than $1$.  Since we assume that multiplication by
$\widetilde D(k)$ provides a continuous transformation of
$\mathcal{S}(\mathbb{R}^d)$ into itself, it is a smooth function all of
whose derivatives are polynomially bounded.  Since the
square root we choose is unambiguous and non-vanishing, the
function $\widetilde D(k)^{1/2}$ also has these properties.

\subsection{Neutral Fields $\Phi$ as Operators}
Define the field $\Phi(f)$, $f\in C^{\infty}_{0}$, as an operator on $\mathcal{E}$ with the domain $\mathcal{D}$ consisting of vectors with a finite number of particles (namely vectors in $\bigoplus_{j=0}^{n} \mathcal{E}_{j}$) and with $C^{\infty}_{0}$ wave functions in each~$\mathcal{E}_{j}$. Note that with our choice of $\sigma$, the field is hermitian as a sesquilinear form, $\Phi(x)=\Phi^{*}(x)$,  if and only if  $\sigma=\overline\sigma$ is real.  This is the case  if and only if the operator  $D=\overline D$ is real.  In  Fourier space, taking the symmetry of $\sigma$ into account,  this is equivalent to  $\overline{\widetilde \sigma (k)} = \widetilde {\sigma}(k)$,  or  $\overline{\widetilde D(k)}=\widetilde D(k)$.   In any case, we assume that  $\sigma$ (or~$D$) is a bounded transformation on $L_{2}(\mathbb{R}^{d})$,
\[
		\| D \| = \| \, |D| \, \|
		= \| D^{*}D \|^{1/2}
		= \| \sigma \|^{2}
		< \infty\;.
\]
	
In fact, it is convenient to assume a stronger bound to ensure good properties for the time-zero fields.  Let us decompose $D=K+iL$ into its real and imaginary parts
\[
		K=\tfrac12\lrp{D+D^{*}}\;,
		\quad \text{and} \quad 
		L=\tfrac{i}{2}\lrp{-D+D^{*}}\;.
\]
 We assume that there are strictly positive constants $M_{1}, M_{2}, M_{3}$ such that for $C=(-\Delta +m^{2})^{-1}$,
\begin{equation}
		M_{1} C \leqslant K \leqslant M_{2}C \;, 
		\quad \text{and} \quad
		\pm K^{-1/2} L K^{-1/2} \leqslant  M_{3}\;.
	\label{D-Bound}
\end{equation}
Now decompose the field into its real and imaginary parts, 
\[
		\Phi(f)
		= c(f) +i d(f)\;,
\]
where $c(f)=\tfrac{1}{2}\lrp{\Phi(f)+\Phi^{*}(f)}$ and $d(f)=-\frac{i}{2}\lrp{\Phi(f)-\Phi^{*}(f)}$.  Let~$\mathfrak{H}_{-\frac{1}{2}}$ denote the Hilbert space of functions on $\mathbb{R}^{d-1}$ with inner product 
\begin{equation}
		\langle f,g \rangle_{\mathfrak{H}_{-1/2}}
		= \langle f, (2\mu)^{-1}g \rangle_{L_{2}(\mathbb{R}^{d-1})}\;,
\end{equation}
 where $\mu=\lrp{-\nabla^{2}+m^{2}}^{1/2}$.

\begin{proposition}\label{Prop:Normal Fields}  
Assume the bounds \eqref{D-Bound}.  Then for $f,g\in L_{2}$ the closure of $\Phi(f)$ is normal and commutes with the closure of $\Phi(g)$.    The fields $c(f)$ and $d(f)$ are essentially self-adjoint on $\mathcal{D}$, and their closures commute.  These results extend to $f=\delta \otimes h$ for $h\in \mathfrak{H}_{-\frac{1}{2}}(\mathbb{R}^{d-1})$.
\end{proposition}

\begin{proof}
The operator $\sigma=(K+iL)^{1/2}$ acting on $L_{2}(\mathbb{R}^{d})$ is invertible as long as $K$ is invertible.  The latter follows from the first bound in the assumption \eqref{D-Bound}.  As a consequence, each vector $\Psi$ in the dense subset $\mathcal{D}\subset\mathcal{E}$ can be represented as  $\Psi=F_{\Psi}\Omega_{0}^{\tt E}$, where $F_{\Psi}$ is a polynomial in fields.  The degree of this polynomial equals the maximum number of particles in the vector.  Thus, for $f\in C^{\infty}_{0}$,
\begin{eqnarray}
		\| \Phi(f)^{n}\Psi  \|_{\mathcal{E}}
		&=& \lra{\Phi(f)^{*\,n}\Phi(f)^{n}\Omega_{0}^{\tt E}, F_{\Psi}^{*}F_{\Psi}\Omega_{0}^{\tt E}}_{\mathcal{E}} ^{1/2}
		\nonumber \\
		&\leqslant&
		\|  {\Phi(f)^{*\,n}\Phi(f)^{n}\Omega_{0}^{\tt E}} \|_{\mathcal{E}}^{1/2}\,
		\| F_{\Psi}^{*}F_{\Psi}\Omega_{0}^{\tt E} \|_{\mathcal{E}}^{1/2}\nonumber \\
		&\leqslant &  (4n-1)!!^{1/4}  \; \; \| \, | D^{1/2}| \, |f| \, \|^{n} \;  \;
		\| F_{\Psi}^{*}F_{\Psi}\Omega_{0}^{\tt E} \|_{\mathcal{E}}^{1/2} \nonumber \\
		&\leqslant &  M2^{n} n!^{1/2} \, \; \| {D} \|_{L_{2}}^{n/2} \, \; \,\| f \|^{n}_{L_{2}}\, \;
		\| F_{\Psi}^{*}F_{\Psi}\Omega_{0}^{\tt E} \|_{\mathcal{E}}^{1/2}  \;.
\end{eqnarray}
Here, $M$ is a  constant. The same bound holds for $\| \Phi(f)^{*\,n}\Psi \|_{\mathcal{E}}$,  and as~$\Phi(f)$ commutes with $\Phi^{*}(f)$ it shows that  $\Psi$ is an analytic vector for~$c(f)$ and for $d(f)$.  One can extend this bound to all $f\in L_{2}(\mathbb{R}^{d})$ by limits, so $\Phi(f)$  is also defined in this case. The essential self-adjointness follows from Nelson's Lemma 5.1 of \cite{Nelson-Analytic Vectors}.  The commutativity of the closures of  $c(f)$ and $d(g)$ follows from a similar estimate for $c(f)^{n} \,d(g)^{n'}\Psi$, showing that the sums 
	\[
		\sum_{n,n'=1}^{N} \frac{i^{n+n'}}{n! n'!}  
		c(f)^{n} \,d(g)^{n'}\Psi
		=  \sum_{n,n'=1}^{N} \frac{i^{n+n'}}{n! n'!}  
		d(g)^{n'} \,c(f)^{n}\Psi
	\]  
converge to $e^{ic(f)}\,e^{id(g)} \Psi = e^{id(g)}e^{ic(f)}\Psi$ as $N\to\infty$.   The vectors $\Psi\in\mathcal{D}$ are dense, so the exponentials commute as unitaries.  The same is true with $f$ and $g$ interchanged. 

In general, use \eqref{D-Bound} to show that 
\begin{eqnarray*}
		|D|
		&=&\lrp{K^{2}+L^{2}}^{1/2}
		=K^{1/2}\lrp{I+K^{-1}L^{2}K^{-1}}^{1/2}K^{1/2}
		\nonumber \\
		&\leqslant&   \lrp{1+M_{3}^{2}}^{1/2}\,K
		\leqslant  \lrp{1+M_{3}^{2}}^{1/2}\,M_{2}  \,  C\;.
\end{eqnarray*}
This  ensures  
\[
		\| \Phi(f)\Omega_{0}^{\tt E} \|_{\mathcal{E}}^{2}
		= \lra{f, |D| \,  f}_{L_{2}}
		\leqslant   \lrp{1+M_{3}^{2}}^{1/2}\,M_{2}  \,   \lra{f, Cf}_{L_{2}} \;.
\]
We use this to restrict to time zero.  If  $f=\delta\otimes h$,   then 
\begin{align*}
		\lra{f, |D| f}_{L_{2}}
		& \leqslant  \lrp{1+M_{3}^{2}}^{1/2}\,M_{2}  \,   \lra{f, Cf}_{L_{2}} \\
		& =  \lrp{1+M_{3}^{2}}^{1/2}\,M_{2}  \,   \lra{h, h}_{\mathfrak{H}_{-1/2}}\;. 		
\end{align*}
The argument then proceeds for the time-zero fields in the same manner  as for the spacetime averaged field.
\end{proof}

\setcounter{equation}{0}
\section{Time Reflection 
\label{Sect:TimeReflection}} Let $\vartheta$ denote time
reflection on $\mathbb{R}^d$, namely $\vartheta\,{:}\, (t,\vec x)
\mapsto (-t,\vec x)$, and also let  $\vartheta$ denote the
implementation of time reflection as a real, self-adjoint,
unitary transformation on $L_2(\mathbb{R}^d)$.  Let
$\Theta=\Theta^*=\Theta^{-1}$ denote the push-forward of
$\vartheta$ to a corresponding time-reflection unitary on $\mathcal{E}$. The
operator $\Theta$ acts on each tensor product $\mathcal{E}_{n}$
as the $n$-fold tensor product of $\vartheta$, and $\Theta\Omega_0^{\tt E}=\Omega_0^{\tt E}$.

\begin{proposition}\label{Prop:TimeReflection Action}
Time-reflection transforms the neutral field $\Phi(x)$ defined by \eqref{Classical_Field-1} according to
\begin{equation}
        \Theta \,\Phi(x)\, \Theta
        = \Phi^{*}(\vartheta x)
        \qquad\text{if and only if}\qquad
        \vartheta \sigma^{\rm T} \vartheta = \sigma^*\;.
    \label{Time-Reflection_field-1}
\end{equation}
In case that \eqref{Time-Reflection_field-1} holds, then both $\vartheta D$ and $D \vartheta$ are self-adjoint operators on $L_{2}(\mathbb{R}^{d})$.  In fact,  
\[
        \vartheta D
        = \sigma^*\vartheta\,\sigma
        \quad\text{and}\quad
         D\vartheta
         = \sigma\vartheta\sigma^*\;.
\]
\end{proposition}

\begin{proof}[\bf Proof]
Since $\vartheta (E,\vec k)=(-E,\vec k)$, the  unitary
$\vartheta$ acts as $(\vartheta f)(k) = f(\vartheta k)$ on
$\mathcal{E}_1$. Thus, $\Theta$ has the property  $\Theta
A(k)\Theta=A(\vartheta k)$, and as a consequence
$\Theta\,\widetilde Q (k)\,\Theta = \widetilde Q
(\vartheta k)$. Using \eqref{Field_Form}, one has
\begin{eqnarray*}
        \Theta\,\Phi(x)\,\Theta
        &=& \lrp{2\pi}^{-d/2}\int \Theta\,\widetilde Q (k)\,\Theta\,
                \widetilde{\sigma}(k)
                    \,e^{ik\cdot x}\,dk
                    \nonumber \\
        &=& \lrp{2\pi}^{-d/2}\int \widetilde Q (\vartheta k)\,
                \widetilde{\sigma}(k)
                    \,e^{ik\cdot x}\,dk 
                    \nonumber \\
        &=& \lrp{2\pi}^{-d/2}\int \widetilde Q ( k)\,
                \widetilde{\sigma}(\vartheta k)
                    \, e^{ik(\cdot\vartheta x)}\,dk
                    \nonumber \\
        &=& \lrp{\lrp{2\pi}^{-d/2}\int \widetilde Q (- k)\,
                \overline{\widetilde{\sigma}(\vartheta k)}
                    \, e^{-ik\cdot(\vartheta x)}\,dk}^*
                    \nonumber \\
        &=& \lrp{\lrp{2\pi}^{-d/2}\int \widetilde Q (k)
                \overline{\widetilde{\sigma}(-\vartheta k)}
                    \,e^{ik\cdot(\vartheta x)}\,dk}^*\;.
\end{eqnarray*}
This equals $\Phi^{*} (\vartheta x)$ if and only if
$\overline{\widetilde{\sigma}(-\vartheta k)}=
\widetilde{\sigma}(k)$. The equivalent condition on
$\sigma$ in configuration space is $\vartheta
\overline{\sigma}\vartheta = \sigma$. Since $\vartheta^2=I$
this is also equivalent to $\vartheta
\sigma^{\rm T}\vartheta=\sigma^*$. As $D$ is given by
\eqref{Neutral Covariance},
\[
        \vartheta D
        = (\vartheta \sigma^{\rm T} \vartheta) \vartheta \sigma
        = \sigma^* \vartheta \sigma\;,
\]
as claimed, and this is clearly self-adjoint. Moreover,
$D\vartheta= \sigma \,\sigma^{\rm T}\vartheta= \sigma
\,\vartheta\sigma^*$ is also self-adjoint.
\end{proof}

\subsection{Time-Reflection Positivity\label{Sect:Time-RP}} Consider the positive-time half-space
$\mathbb{R}^d_{+}=[0,\infty)\times\mathbb{R}^{d-1}\, $; the negative-time half-space  $\mathbb{R}^{d}_{-}$ is defined similarly. Let $L_2(\mathbb{R}^d_\pm)$ denote
the subspace of~$L_2(\mathbb{R}^d)$ consisting of functions
supported in $\mathbb{R}^d_\pm$. Let $\mathcal{E}_{\pm,0}\subset \mathcal{E}$
denote the subspace of finite linear combinations $A=\sum_{j=1}^N c_j
\,e^{i\Phi(f_j)}\,\Omega_0^{\tt E}$, for $f_j\in\mathcal{S}(\mathbb{R}^d_\pm)$, and
let $\mathcal{E}_\pm$ denote its closure in $\mathcal{E}$. The
Osterwalder-Schrader reflection form on
$\mathcal{E}_{+,0}\times \mathcal{E}_{+,0}$ (or, alternatively, on $\mathcal{E}_{-,0}\times \mathcal{E}_{-,0}$) is
\begin{equation}
        \lra{A,B}_\mathcal{H}
        = \lra{A, \Theta B}_\mathcal{E}\;.
    \label{OS-Form-1}
\end{equation}
This extends by continuity to $\mathcal{E}_\pm\times \mathcal{E}_\pm$. Thus, the left hand side extends to a pre-inner product on  equivalence classes
$[A] = A+ n \in \mathcal{E}_{\pm}$. Here,  $n\in\mathcal{E}_{\pm}$  is an element of the null space of the 
form~$\lra{\ \cdot \ , \ \cdot \ }_\mathcal{H}$.  For simplicity denote the equivalence class $[A]$ by $A$.  The space $\mathcal{H}$ is the completion of the equivalence classes in this inner product. Let $M$
denote the $N\times N$ matrix with entries
\begin{equation}
        M_{jj'}=S(f_{j'}-\vartheta \overline {f_{j}})\;,
        \quad \text{where} \quad 
        f_{j}\in \mathcal{E} 
        \quad \text{for} \quad
        j=1,\ldots, N\;.
    \label{Matrix M}
\end{equation}

\begin{definition} [\bf RP] \label{Def:Neutral-OS Positivity}
Various possible formulations of RP (Osterwalder-Schrader Positivity)  with respect to time reflection
are:
    \begin{enumerate}
    
    \item[(i)] \label{RP Phi} One has $0\leqslant
   \Theta $ on either $\mathcal{E}_{+}$ or $\mathcal{E}_{-}$.    If both hold, then $\Theta$ is doubly RP.

    \item[(ii)]\label{RP S(f)} The functional $S(f)$ is RP with respect to $\vartheta$, if for every choice of 
    $N \in {\mathbb N}$ and for all functions $f_j\in \mathcal{S}(\mathbb{R}^d_{+})$ the matrix $M$ in~\eqref{Matrix M} is positive
    definite.  The functional $S(f)$  is also RP, if $M$ is positive for $f_{j}\in\mathcal{S}(\mathbb{R}^{d}_{-})$.  If both conditions hold, then $S(f)$  is {\em doubly}-reflection positive.

    \item[(iii)] \label{RP D}
    A symmetric operator $D=D^{{\rm T}}$ on $L_2(\mathbb{R}^d)$ is RP  with respect to $\vartheta$, if $ 0\leqslant \vartheta D$ on $ L_2(\mathbb{R}^d_+)$.  It is also RP if $0\le\vartheta  D$ on $L_{2}(\mathbb{R}^{d}_{-})$. The latter is equivalent to $0\leqslant D\vartheta   $ on $L_{2}(\mathbb{R}^{d}_{+})$.  The 
operator~$D$ is doubly RP if both conditions hold.
    \end{enumerate}
\end{definition}

\begin{proposition}\label{Prop:RP-Time Equivalence} Let
$\Phi$ be a field on $\mathcal{E}$ defined by
\eqref{Classical_Field-1}, and assume \eqref{Time-Reflection_field-1}.
Then
	\begin{itemize} 
	\item[(i)]  The characteristic functional S(f) defined in
\eqref{Characteristic_Functional on E} satisfies
$S(f)=\overline{S(-\vartheta \overline f)}$.  
	\item[(ii)]  The matrix
$M=M^*$ defined in \eqref{Matrix M} is hermitian.
	\item[(iii)]  Statements III.2.(i) and III.2.(ii) of Definition~\ref{Def:Neutral-OS Positivity} are equivalent.   
	\item[(iv)]  All three Statements  III.2.(i),   III.2.(ii), and   III.2.(iii) of  Definition~\ref{Def:Neutral-OS Positivity} are equivalent in case that   
\begin{equation}
        S(f)
        =e^{-\tfrac{1}{2}\lra{\bar f, Df}_{L_2}}\;,
        \quad \text{with} \quad
        D=\sigma\sigma^{\rm T}=D^{\rm T}\;.
    \label{Gaussian Characteristic Function}
\end{equation}
\end{itemize}
\end{proposition}

\begin{proof}[\bf Proof]
The equivalence of Statements III.2.(i) and III.2.(ii) follows from  the identity
\begin{equation}
        \lra{A,A}_{\rm OS}
        = \sum_{j,j'=1}^N
        \overline{c_j}\,c_{j'}\,S(f_{j'}-\vartheta\overline { f_{j}})
        =    \sum_{j,j'=1}^N
        \overline{c_j}\,c_{j'}\,M_{jj'}\;,
    \label{S-Form RP}
\end{equation}
which is a consequence of 
\begin{eqnarray*}
        \lra{A,A}_{\rm OS}
        &=& \sum_{j,j'=1}^N \overline{c_j}\,c_{j'}\,  
        \lra{e^{i\Phi(f_{j})}\,\Omega_{0}^{\tt E}, e^{i\Phi(f_{j'})}\,\Omega_{0}^{\tt E}}_\mathcal{E}      
        \nonumber \\
        &=& \sum_{j,j'=1}^N \overline{c_j}\,c_{j'}\,
            \lra{\Omega_0^{\tt E}
                    \,,e^{-i\Phi^*(\overline{f_j})}\,
                   \Theta\,e^{i\Phi(f_{j'}\,)}\Omega_0^{\tt E}}_\mathcal{E}\;.
\end{eqnarray*}
As $\Theta\Omega_0^{\tt E}=\Omega_0^{\tt E}$, one can substitute
$\Theta\,
e^{-i\Phi^*(\overline{f_j})}\,\Theta=e^{-i\Phi(\vartheta\overline{
f_j})}$ to give \eqref{S-Form RP}.

To verify that $S(f)=\overline{S(-\vartheta \overline f)}$ is valid, compute 
\begin{eqnarray*}
        \overline{S(- \vartheta \overline{f})}
        &=& \overline{\lrae{\Omega_0^{\tt E}}{e^{-i\Phi(\vartheta \overline f)}}}_\mathcal{E}
        = \lra{e^{-i\Phi(\vartheta \overline f)}\Omega_0^{\tt E}, \Omega_0^{\tt E}}_\mathcal{E}
        \nonumber \\
        &=& 
	\lra{\Omega_0^{\tt E}, e^{i\Phi^*(\vartheta f)}\Omega_0^{\tt E}}_\mathcal{E}
	=        \lra{\Omega_0^{\tt E}, \Theta e^{i\Phi^*(\vartheta f)}\Theta\Omega_0^{\tt E}}_\mathcal{E}
                \nonumber \\
        &=& 
	\lra{\Omega_0^{\tt E}, e^{i\Phi( f)}\Omega_0^{\tt E}}_\mathcal{E}
        = S(f)\;,
\end{eqnarray*}
Using this relationship one also sees that the matrix $M$ is hermitian, for  
\[
        M_{jj'}
        = S(f_{j'}- \vartheta\overline{ f_{j}})
        = \overline{S(-\vartheta \overline{f_{j'}}+  f_{j})}
        = \overline{M_{j'j}}\;.
\]

Finally consider the  Gaussian characteristic function \eqref{Gaussian Characteristic Function}.  Now we show the equivalence of 
Statements III.2.(ii) and  III.2.(iii). Take $A\in\mathcal{E}_{+,0}$ and choose 
$N=2$,  $f_1=-i\lambda f$, and $f_2=0$. Then for 
\[ c_1=\lambda^{-1}=-c_2 \; , \]
one can take the limit $\lambda\to0$ of $0\leqslant\lra{A,A}_{\mathcal{H}}$.  This is just 
\begin{eqnarray*}
       0
       & \leqslant &\lra{\Phi(f)\Omega_0^{\tt E}\,,\Theta\Phi(f)\Omega_0^{\tt E} }_\mathcal{E}
       = \lra{\Omega_0^{\tt E}\,,\Phi^*(\overline f)\,
                \Theta\Phi(f)\Omega_0^{\tt E} }_\mathcal{E}
      \nonumber \\
      & = &\lra{\Omega_0^{\tt E}\,,\Phi(\vartheta\overline  f)\,
                \Phi(f)\Omega_0^{\tt E} }_\mathcal{E}
       = \lra{f,\vartheta D f}_{L_2} \;.
\end{eqnarray*}
Hence, Statement III.2.(ii)  ensures $0\leqslant\vartheta D$ on $ L_2(\mathbb{R}^d_+)$.   

The converse is also true.  To see this, use 
	\[
        		\sum_{j,j'=1}^N \overline {c_j} \,c_{j'}\, M_{jj'}
		= \sum_{j,j'=1}^N \overline {c_j} \,c_{j'}\,
                S(f_{j'}-\overline {\vartheta f_{j}})
		= \sum_{j,j'=1}^N \overline {d_j} \,d_{j'}\, e^{\lra{f_{j}, \vartheta D f_{j'}}_{L_{2}}}\;,
	\]
where
	\[	
		d_{j}=c_{j}e^{-\tfrac{1}{2}\lra{\overline {f_{j}}, Df_{j}}_{L_{2}}}\;.
	\label{PositiveS-Form}
	\]
In fact, 
\begin{eqnarray}
        &&\sum_{j,j'=1}^N \overline {c_j} \,c_{j'}\,
                S(f_{j'}-\overline {\vartheta f_{j}})
        = 
	\sum_{j,j'=1}^N \overline {c_j} \,c_{j'}\,
                e^{-\tfrac{1}{2}
                    \lra{\overline{f_{j'}} - \vartheta f_j \,,
                    D( f_{j'} -\vartheta\overline{ f_j})}_{L_{2}}}\nonumber \\
        &&\qquad \qquad
        = \sum_{j,j'=1}^N
           { \lrp{ \overline {c_j} e^{-\tfrac{1}{2}
                    \lra{\vartheta f_j, D \vartheta\overline{f_j}}_{L_{2}}}}}
                \lrp{c_{j'}\,e^{-\tfrac{1}{2} \lra{\overline{f_{j'}}\,,
                        D  f_{j'}}_{L_{2}}}} \times
                        \nonumber \\
          &&  \qquad   \qquad\qquad\qquad         \times
                e^{\tfrac{1}{2} \lra{f_j \,, \vartheta Df_{j'}}_{L_{2}}
                + \tfrac{1}{2}\lra{\overline{f_{j'}}
                        \,, D \vartheta\overline{ f_j}}_{L_{2}}}\nonumber \\
        &&\qquad \qquad
        = \sum_{j,j'=1}^N \overline {d_j}\,d_{j'}\,
            e^{\tfrac{1}{2}\lra{{f_{j}}\, ,\,
              \vartheta(D+D^{\rm T})\,{f_{j'}}}_{L_{2}}}
              \nonumber \\
        &&\qquad \qquad
        = \sum_{j,j'=1}^N \overline {d_j}\,d_{j'}\,
            e^{\lra{f_j\, , \,\vartheta D f_{j'}}_{L_{2}}}
        \;.
    \label{RP-Gaussian Identity}
\end{eqnarray}
The second to last equality follows from using $\vartheta D^{*}\vartheta = D$ in 
    \[
        \lra{\vartheta f_j,D \vartheta\overline{ f_{j}}}_{L_{2}}
        = \overline{\lra{D\vartheta \overline{f_{j}}, \vartheta f_j}}_{L_{2}}
        = \overline{\lra{\vartheta\overline{f_{j}}, D^*\vartheta f_j}}_{L_{2}}
        = \overline{\lra{\overline{ f_{j}}, D f_j}}_{L_{2}}\;,
    \]
and from 
\begin{eqnarray*}
        \lra{\overline{f_{j'}}
                        \,, D\vartheta\overline{ f_j}}_{L_{2}}
        &=& \overline{\lra{D\vartheta\overline{f_j}
                        \,, \overline{f_{j'}} }}_{L_{2}}
        = \overline{\lra{\vartheta\overline{ f_j}
                        \,,  D^*\overline{f_{j'}} }}_{L_{2}}
        \nonumber \\
        &=&\lra{\vartheta f_j,D^{\rm T} f_{j'}}_{L_{2}}
        =\lra{f_j,\vartheta D^{\rm T} f_{j'}}_{L_{2}}\;.
\end{eqnarray*}
The last inequality in \eqref{RP-Gaussian Identity} then is a consequence of the symmetry of $D$.
    
Thus, Statement  III.2.(ii) on $\mathcal{E}_{+}$ follows from the positivity of the matrix 
$\mathcal{K}$
with entries 
\[ 
\mathcal{K}_{jj'} = e^{\lra{f_{j},\vartheta D f_{j'}}_{L_{2}}}\quad \text{ for $f_{j}\in L_2(\mathbb{R}^d_+)\, $}.
\]
 In fact, we now see that $I\leqslant  \mathcal{K}$. The assumption $0\leqslant  \vartheta D$ on
$L_{2}(\mathbb{R}^d_+)$ means that the matrix ${\mathbf k}$ with entries
${\mathbf k}_{jj'}=\lra{f_j, \vartheta D f_{j'}}_{L_{2}}$ has non-negative
eigenvalues for $f_j\in L_2(\mathbb{R}^d_+)$.  The same is true for
the matrix ${\mathbf k}^{\circ n}$ with entries ${\mathbf k}^{\circ n}_{jj'}=({\mathbf k}_{jj'})^n$.  In fact, the matrix~${\mathbf k}^{\circ n}$ equals~${\mathbf k}^{\otimes \,n}$ on the diagonal, and the
eigenvalues $\lambda_j(\mathcal{K})$ of the matrix 
$\mathcal{K}=\sum_{n=0}^\infty n!^{-1} {\mathbf k}^{\circ n}$ with entries
$\mathcal{K}_{jj'}=e^{{\mathbf k}_{jj'}}$ equal $1$ plus an eigenvalue of a
positive matrix. Consequently, $1\leqslant \lambda_j(\mathcal{K})$, as claimed.  

A function $f\in L_2(\mathbb{R}^d_-)$, if and only if $\vartheta f\in L_2(\mathbb{R}^d_+)$.  Thus, Statement~\ref{RP S(f)} on $\mathcal{E}_{-}$ is equivalent to  $0\leqslant\vartheta^{2} D \vartheta=D\vartheta$ on $L_2(\mathbb{R}^d_+)$.
\end{proof}
\goodbreak

\setcounter{equation}{0}
\section{Quantization\label{Sect:Quantization}}
If the time reflection  $\Theta$ is reflection positive (or doubly reflection positive), then the  form  \eqref{OS-Form-1} defines a pre-Hilbert space (or two pre-Hilbert spaces) $\mathcal{H}_{\pm, 0}$ whose elements are equivalence classes 
\[ \widehat A=\{A+n\}\; , \]
where $A,n\in\mathcal{E}_{+}$ or $A,n\in\mathcal{E}_{-}$ and where  $n$ is an element of the null space of the form~\eqref{OS-Form-1}.  In either case, the inner product on $\mathcal{H}_{\pm, 0}$ is given by the form \eqref{OS-Form-1}.  Let $\mathcal{H}_{\pm}$ denote the completions of the pre-Hilbert spaces $\mathcal{H}_{\pm, 0}$.

In case the characteristic functional $S(f)$ is both time-translation invariant and reflection positive, one obtains a positive Hamiltonian operator from the RP inner product.  This follows from the standard construction provided in \cite {Osterwalder-Schrader}, so we refer to this operator as the {\em OS-Hamiltonian}.  Let the time translation group $T(t)$ act on functions as 
\[
		(T(s)f)(t, \vec x)
		= f(t-s, \vec x)\;.
\]
If $S(T(s)f)=S(f)$ for all $s$, then $S(f)$ is time-translation invariant.  In this case, $T(t)$ acts as a unitary transformation on $\mathcal{E}$, and 
\[
		T(t) \colon \mathcal{E}_{\pm} \to \mathcal{E}_{\pm}\; \text{for } 0\leqslant  \pm t\;.
\]

\begin{proposition} [\bf OS-Hamiltonian] \label{Prop:OSHamiltonian} Let $^{\wedge}$ denote the canonical projection from $\mathcal{E}_{\pm}$ to $\mathcal{H}_{\pm}$ resulting from the reflection-positive OS form.  Then the maps 
\[
		R(t)_{\pm} = \widehat{ T(t)} \quad \text{acting on $\mathcal{H}_{\pm}$}
\]
are contraction semigroups with infinitesimal generators $H_{\pm} \, $, namely $R_{\pm}(t)=e^{-tH_{\pm}}$ with $0 \leqslant  H_{\pm} \, $.
\end{proposition}

The proof of this proposition follows from the arguments given 
in~\cite{Osterwalder-Schrader}. 
In the Gaussian case, let~${\boldsymbol{h}}_{\pm}$ denote the restriction of $H_{\pm}$ to the one-particle subspace $\mathcal{H}_{1,\pm}\subset \mathcal{H}_{\pm}$.  If the functional $S(f)$ has  covariance~$D$, the reflection positivity condition on the one-particle space requires $0\leqslant  \vartheta D$ on~$L_{2}(\mathbb{R}^{d}_{\pm})$. For example, if $f\in L_{2}(\mathbb{R}^{d}_{+})$, then  
\[
		\lra{f, \vartheta Df}_{L_{2}}
		= \lra{F,F}_{\mathcal{H}_{1,+}}\;,
\]
where $F=\int_{0}^{\infty} e^{-t\boldsymbol{h}_{+}} f_{ t} \,dt$ and $f_{t}(\vec x)=f(t,\vec x)$.  Likewise, if $f\in L_{2}(\mathbb{R}^{d}_{-})$, then one has
\[
		\lra{f, \vartheta Df}_{L_{2}}
		= \lra{F,F}_{\mathcal{H}_{1,-}}\;
		\quad \text{for} \quad
		F=\int_{-\infty}^{0} e^{t\boldsymbol{h}_{-}} f_{ t} \,dt\;.
\]

\subsection{Spatial Reflection-Positivity for the Neutral Field}
In a fashion similar to time-reflection, one considers
spatial reflections. The spatial reflection through a plane
orthogonal to a given spatial vector $\vec n\in\mathbb{R}^{d-1}$ is 
$\pi_{\vec n} \colon (t,\vec x)\mapsto (t, \vec x -2 \lrp{\vec
n\cdot\vec x}\vec n)$. Let~$\pi_{\vec n}$ denote
the action that this reflection induces as a real,
self-adjoint, unitary on~$L_2(\mathbb{R}^d)$.

Let $\Pi_{\vec n}$ denote the push-forward of $\pi_{\vec
n}$ to a real, self-adjoint unitary on $\mathcal{E}$. Instead of
the positive-time subspace $L_2(\mathbb{R}^d_+)$ that arose in
the study of time reflection in \S\ref{Sect:TimeReflection}, 
use the subspace of functions $L_2(\mathbb{R}^d_{\vec
n\pm})$ supported on one side of the reflection hyperplane.
Let $\mathcal{E}_{\vec n\pm}\subset \mathcal{E}$ denote the closure of the subspace
spanned by the vectors 
\[ \left\{e^{i\Phi(f)}\Omega_0^{\tt E}\right\} \quad \text{with
$f\in\mathcal{S}(\mathbb{R}^d_{\vec n\pm})$.} 
\]
Spatial reflection $\Pi_{\vec
n}$  leaves $\Omega_0^{\tt E}$ invariant and maps $f(E, \vec
k)\in \mathcal{E}_1$ to $f(\pi_{\vec n} k)$. In
Fourier space 
\[ \pi_{\vec n}(E,\vec k)=(E,\vec k-2(\vec
k\cdot\vec n) \vec n)\] 
and~$\pi_{\vec n}^2 k=k$.
Spatial-reflection transforms the $\widetilde Q(k)$'s according to
$\Pi_{\vec n} \,\widetilde Q(k) \Pi_{\vec n}= \widetilde Q(\pi_{\vec n}k)$.
Following the proof of Proposition~\ref{Prop:TimeReflection Action},
we obtain:

\begin{proposition}[\bf Spatial Reflection of the Field]
\label{Prop:RP-Spatial Equivalence} Let the neutral field
$\Phi(x)$ be defined by~\eqref{Classical_Field-1} on
$\mathcal{E}$. Then the spatial-reflection $\Pi_{\vec n}$
transforms $\Phi$ according to
\begin{equation}
        \Pi_{\vec n} \,\Phi(x)\, \Pi_{\vec n} 
        = \Phi^{*} (\pi_{\vec n}\, x)\;,ƒ
    \label{Spatial-Reflection_field-1}
\end{equation}
if and only if
\begin{equation}
        \pi_{\vec n} \,\sigma^{\rm T}\, \pi_{\vec n} = \sigma^*\;.
    \label{SpatialReflectionSigma}
\end{equation}
If \eqref{SpatialReflectionSigma} holds, then
    $
        \pi_{\vec n} D
        = \sigma^* \pi_{\vec n} \sigma
        = D^* \pi_{\vec n} \,
    $,
\[
        \pi_{\vec n} D
        = \lrp{\pi_{\vec n} D}^*\;,
        \quad\text{and}\quad
         D\pi_{\vec n}
         = \lrp{D\pi_{\vec n}}^*\;
\]
are self-adjoint operators on $L_2(\mathbb{R}^d)$.
\end{proposition}

One can formulate spatial reflection positivity by
substituting $\pi_{\vec n}$ for~$\vartheta $ and $\Pi_{\vec
n}$ for~$\Theta$, and $\mathcal{E}_{\vec n\pm}$ for $\mathcal{E}_\pm$ in
Definition~\ref{Def:Neutral-OS Positivity}.  Spatial
reflection positivity gives a pre-inner product
\begin{equation}
		\lra{A,B}_{\mathcal{H}(\vec n)}
		= \lra{A,\Pi_{\vec n}\,B}_{\mathcal{E}}
\end{equation}
on $\mathcal{E}_{\vec n+}$, as well
as  new requirements on the transformation of the
field or on $D$. The proof of the following result is identical to
the proof of Proposition~\ref{Prop:RP-Time Equivalence}.

\begin{proposition}[\bf Spatial-Reflection Positivity]
If one replaces $\vartheta$, $\Theta$, $\mathbb{R}^d_\pm$, $\mathcal{E}_\pm$ in the
statement of Proposition~\ref{Prop:RP-Time Equivalence} by
$\pi_{\vec n}, \Pi_{\vec n}, \mathbb{R}^d_{\vec n\pm}, \mathcal{E}_{\vec
n\pm}$, respectively, then the proposition remains valid.
\end{proposition}

\setcounter{equation}{0}
\section{Charged  Fields
\label{Sect:Charged Fields}}
One defines charged fields $\Phi_\pm(x)$ to replace the usual hermitian-conjugate fields $\Phi(x)$ and $\Phi^{*} (x)$.\footnote{Note that we use ``$\pm$'' here to label charges, while in earlier sections we use this notation to label positive and negative time subspaces, \emph{etc.}  We hope that this causes no confusion.}  These fields are linear in the complex coordinates $\widetilde Q_\pm(k)$ and the kernels $\sigma_{\pm}$, whose properties we elaborate on below. Similar to \eqref{Classical_Field-1}, take the charged fields to be
\begin{equation}
        \Phi_\pm(x)
        = \lrp{2\pi}^{-d/2}
            \int \widetilde Q_\pm(k)\, \widetilde \sigma_\pm(k)
             \,e^{ik\cdot x}\,dk\;.
    \label{Classical_Charged Field-1}
\end{equation}

While one might regard the charged coordinates $\widetilde Q_{\pm}$, or their Fourier transforms $Q_{\pm}$, as linear
combinations of two independent sets of neutral coordinates $ Q_{\pm}=\frac{1}{\sqrt{2}}\lrp{Q_{1}\pm i Q_{2}}$, here we do not require this. We take independent creation and annihilation operator-valued distributions $A^{*}_\pm(k)$ and $A_\pm(k)$, which act on a 
Fock-Hilbert space ${\boldsymbol {\mathcal E}}$ and satisfy the relations
    \[
      [ A_\pm(k) , A_{\pm'}^{*} (k') ] =
       \delta_{\pm \pm'}\,\delta(k-k')\;,
       \quad \text{and} \quad
       [A_\pm(k) , A_{\pm'}(k')]=0\;.
    \]
One can decompose the  one-particle space into a direct sum 
\[ 
	{\boldsymbol {\mathcal E}}_{1} = \mathcal{E}_{1}^{+} \oplus \mathcal{E}_{1}^{-} \; .  
\]
Correspondingly, the Fock space ${\boldsymbol {\mathcal E}}$ decomposes into the tensor product\footnote{The spaces ${{\boldsymbol {\mathcal E}}}^{\pm}$  labelled with a superscript denote  spaces where each particle has positive or negative charge.  The spaces   ${{\boldsymbol {\mathcal E}}}_{\pm}$ are subspaces of vectors at  positive or negative times.  However, for the fields themselves, and  some other associated operators, we retain the notation used elsewhere in this section to label charges by subscripts.}
 ${\boldsymbol {\mathcal E}}={\boldsymbol {\mathcal E}}^{+}\otimes{\boldsymbol {\mathcal E}}^{-}$.

One defines the complex coordinates as 
\[
        \widetilde Q_+(k)
        = A^*_+(k) + A_-(-k)\;,
 \]
and         
\begin{equation}
        \widetilde Q_-(k)
        = A^*_-(k) + A_+(-k)\;,
        \quad \text{so} \quad
        \widetilde Q_{-}(k)
        = \widetilde Q_+^{*} (-k)\;.
    \label{Chraged Coordinates Q}
\end{equation}
These coordinates all mutually commute,
\[
        [ \widetilde Q _\pm(k) , \widetilde Q _{\pm'}(k') ]
        =0
        = [ \widetilde Q _\pm(k) , \widetilde Q _{\pm'}^{*} (k')] 
        \;.
\]
As a consequence, the fields and their
adjoints commute 
\[
        [ \Phi_\pm(x) , \Phi_{\pm'}(x')]
        = [ \Phi_{\pm}(x) , \Phi_{\pm'}^*(x') ]=0\;.
\]
The characteristic function of the charged field depends on two variables  
\[
        S(f_{+},f_{-})
        = \lrae{{\boldsymbol \Omega}_0^{\tt E}}{e^{i\Phi_+(f_{+})\,+\,i\Phi_-(f_{-})}\,}_{\boldsymbol {\mathcal E} }\;.
 \]

\subsection{Twist Symmetry of the Charged Field}
Define two number operators
    \[
        N_\pm = \int A_\pm^{*} (k) A_\pm(k)\,dk\;.
    \]
In terms of these, the total number operator is  $N=N_{+}+N_{-}$.  The charge (or vortex number) is
$F=N_+-N_-$.  The vortex number $F$ implements the twist transformation on the field:  
    \[
        e^{i\alpha F}\widetilde Q _\pm(k) e^{-i\alpha F}
        = e^{\pm i\alpha}\,\widetilde Q _\pm(k)\;, \qquad \alpha\in\mathbb{C} \; .
    \]
As a consequence,
\[
        e^{i\alpha F}\Phi_\pm(x) e^{-i\alpha F}
        = e^{\pm i\alpha} \Phi_\pm(x)\;.
\]
Moreover, the zero-particle vector is twist invariant, $e^{i\alpha F}\Omega_{0}^{\tt E}=\Omega_{0}^{\tt E}$.
Unitarity of $e^{i\alpha F}$ then ensures the vanishing of the diagonal  expectations
  \[
        \lrae{{\boldsymbol \Omega}_0^{\tt E}}{\Phi_+(x) \,\Phi_+(x')\,}_{{\boldsymbol {\mathcal E}}}
        = \lrae{{\boldsymbol \Omega}_0^{\tt E}}{\Phi_-(x) \,\Phi_-(x')\,}_{{\boldsymbol {\mathcal E}}}
        =0\;.
    \]

The remaining two-point functions do not vanish. Let 
\[
        D(x-x')
        = (2\pi)^{-d}\,
        \int \widetilde
        \sigma_+(k)\,\widetilde
        \sigma_-(-k)\,e^{ik\cdot(x-x')}dk\;,
\]
be the kernel of the operator
\begin{equation}
        D = \sigma_+\sigma_-^{\rm T}
        = \sigma_-^{\rm T}\sigma_+
        \quad \text{with transpose }\quad
        D^{\rm T}
        =\sigma_-\sigma_+^{\rm T}
         =\sigma_+^{\rm T}\sigma_-\;.
    \label{Charged Covariance}
\end{equation}
Then  
\[
        \lrae{{\boldsymbol \Omega}_0^{\tt E}}{\Phi_+(x) \,\Phi_-(x')\,}_{{\boldsymbol {\mathcal E}}}
        = D(x-x')
\]
and
\[
        \lrae{{\boldsymbol \Omega}_0^{\tt E}}{\Phi_-(x) \,\Phi_+(x')\,}_{{\boldsymbol {\mathcal E}}}
        = D^{\rm T}(x-x')\;.
    \label{-Expectation}
\]
As
$\sigma_+, \sigma_-^{\rm T}$ are both translation invariant,
they commute. Using these relations, the characteristic function of the charged field is
\begin{eqnarray}
        S(f_+,f_-)
        &=& \sum_{n=0}^\infty \frac{(-1)^n}{n!^2}\,
                \lrae{{\boldsymbol \Omega}_0^{\tt E}}{\Phi_+(f_+)^n\,\Phi_-(f_-)^{n}\,}_{{\boldsymbol {\mathcal E}}}
       \nonumber \\
       & = & e^{-\lra{\overline f_+, Df_-}_{L_2}}\;.
    \label{Characteristic_Functional_Charged}
\end{eqnarray}
Finally, note that
\begin{equation}
        D=D^*\;,
        \quad \text{either if} \quad \sigma_\pm = (\sigma_\pm)^* \quad \text{or if} \quad
        \sigma_\pm =\overline{\sigma_\mp}\;.
    \label{D_Self-Adjoint}
\end{equation}
Also
\[
        D=D^{\rm T}\;,
        \text{either if}\ \sigma_+ = \sigma_- \;\ \ \text{or if}
        \sigma_\pm =\sigma_\pm^{\rm T}\;.
\]
Thus there are in principle four ways that
$D=D^*=D^{\rm T}=\overline D$.  They are: (i) $\sigma_+=\sigma_-=\sigma_+^*=\sigma_-^*$, or (ii) $\sigma_+=\sigma_+^{\rm T}=\overline{ \sigma_+}$
along with  $\sigma_-=\sigma_-^{\rm T}=\overline {\sigma_-}$, or (iii)
$\sigma_+=\sigma_-=\overline{\sigma_+}$, or (iv) $\sigma_+=\sigma_-=\sigma_+^{\rm T} $.

\subsection{Matrix Notation for the Charged Field}
We combine the two components $\Phi_{\pm}$ of the charged field into a vector 
	\[
		{\boldsymbol \Phi}= \begin{pmatrix}
		\Phi_{+}\\
		\Phi_{-}
		\end{pmatrix}
		\text{that pairs with test functions }
		{\boldsymbol f} = 
		\begin{pmatrix}
		f_{+}\\
		f_{-}\end{pmatrix}\;,
	\] 
to yield a field acting on ${\boldsymbol {\mathcal E}}$, 
    \[
        {\boldsymbol \Phi}({\boldsymbol f}) = \sum_{\alpha=\pm}
        \Phi_\alpha(f_\alpha)\;.
    \]
Likewise, combine the various two-point functions of the
charged field into a $2\times 2$ matrix of operators~$\boldsymbol{D}$ with entries indexed by $\alpha=\pm$, each acting on $L_2(\mathbb{R}^d)$. Let $\boldsymbol{L}_2=L_{2}(\mathbb{R}^{d})\oplus L_{2}(\mathbb{R}^{d})$, and let~$\overline {\boldsymbol f}$ denote complex conjugation of ${\boldsymbol f}$.  Then, with $D$ defined by \eqref{Charged Covariance}, let 
    \[
        \boldsymbol{D}
        = \bm{0}{D}{D^{\rm T}}{0}\;.
    \]
One observes that $\boldsymbol{D}$ is symmetric,  $\boldsymbol{D}=\boldsymbol{D}^{\rm T}$.  The two-point function for the field ${\boldsymbol \Phi}$  is
    \[
        \lrae{{\boldsymbol \Omega}_0^{\tt E}}{{\boldsymbol \Phi}({\boldsymbol f}) {\boldsymbol \Phi}(\boldsymbol{g})}_{{\boldsymbol {\mathcal E}}}
        = \lra{\overline{\boldsymbol f}, \boldsymbol{D}\, \boldsymbol{g} }_{\boldsymbol{L}_2}
        = \lrae{{\boldsymbol \Omega}_0^{\tt E}}{{\boldsymbol \Phi}(\boldsymbol{g}){\boldsymbol \Phi}({\boldsymbol f})}_{{\boldsymbol {\mathcal E}}}\;.
    \]
The second identity arises from the commutativity of the field.  Equivalent to commutativity is the relation
    \[
        \lra{\overline{\boldsymbol{g}}, \boldsymbol{D}\,{\boldsymbol f}}_{\boldsymbol{L}_2}
        = \lra{\boldsymbol{D}^*\overline{\boldsymbol{g}} , \,{\boldsymbol f}}_{\boldsymbol{L}_2}
        = \lra{\overline{\boldsymbol{D}\boldsymbol{g}}, \,{\boldsymbol f}}_{\boldsymbol{L}_2}
        = \overline{\lra{\,{\boldsymbol f},\overline{\boldsymbol{D}\boldsymbol{g}}}}_{\boldsymbol{L}_2}
        = \lra{\overline{\boldsymbol f}, \boldsymbol{D}\, \boldsymbol{g}}_{\boldsymbol{L}_2}\;.
    \label{Conjugate_InnerProductL2}
    \]
Note also that 
\[
		\lra{\overline {\boldsymbol f} , \boldsymbol{D}\, \boldsymbol{g} }_{{\boldsymbol{L}_2}}
		= \lra{\overline {f_{+}}, D g_{-}}_{L_{2}}
			+ \lra{\overline {g_{+}}, D f_{-}}_{L_{2}}\;,
\]
so 
\begin{equation}
		\lra{\overline {\boldsymbol f} , \boldsymbol{D}\, {\boldsymbol f}}_{{\boldsymbol{L}_2}}
		= 2\lra{\overline {f_{+}}, D f_{-}}_{L_{2}}
		\;.
	\label{Two-Point-Charged-Matrix}
\end{equation}
In fact,    
\[
        \lra{\overline{f_-}, D^{\rm T} g_+}_{L_{2}}
        = \overline{\lra{ D^{\rm T} g_+, \overline{f_-}}}_{L_{2}}
        = \lra{ D^{*}\overline {g_+},  f_-}_{L_{2}}
        = \lra{ \overline{g_+}, D f_-}_{L_{2}}\;.
\]

\begin{proposition}
The characteristic function $S(f_+,f_-)$ of
\eqref{Characteristic_Functional_Charged} has  the standard form of a Gaussian,
\begin{eqnarray}
        S(\mathbf f)
        &=&\lrae{{\boldsymbol \Omega}_0^{\tt E}}{e^{i\mathbf{\Phi(f)}}}_{{\boldsymbol {\mathcal E}}}
        = e^{-\frac{1}{2}\lra{\overline {\boldsymbol{f}}, \boldsymbol{D f}}_{{\boldsymbol{L}_2}}}
        \nonumber
       \\
       &=& S(f_+,f_-)
        = e^{-\lra{\overline{f_+},Df_-}_{L_2}}\;.
    \label{Characteristic_Functional on E-2}
\end{eqnarray}
\end{proposition}

\begin{proof}[\bf Proof]
The product of $n \in {\mathbb N}$ fields with test functions ${\boldsymbol f}^{(j)}$
with components $f_\alpha^{(j)}$, $j=1, \ldots, n$, is
\begin{equation}
    \label{n-Field Monomial}
        {\boldsymbol \Phi}({\boldsymbol f}^{(1)})\cdots{\boldsymbol \Phi}({\boldsymbol f}^{(n)})
        = \sum_{\alpha_1,\ldots,\alpha_n=\pm}
        \Phi_{\alpha_1}(f_{\alpha_1}^{(1)})\cdots
                \Phi_{\alpha_n}(f_{\alpha_n}^{(n)})\;.
\end{equation}
As the fields are linear in creation and annihilation
operators, the zero-particle expectation of such a product
satisfies the Gaussian recursion relation
    \begin{multline}
        \lrae{{\boldsymbol \Omega}_0^{\tt E}}{{\boldsymbol \Phi}({\boldsymbol f}^{(1)})\cdots{\boldsymbol \Phi}({\boldsymbol f}^{(n)})\,}_{{\boldsymbol {\mathcal E}}}
        = \sum_{j=2}^n  \lra{\overline{{\boldsymbol f}^{(1)}},
        \boldsymbol{D}\,{\boldsymbol f}^{(j)}}_{{\boldsymbol{L}_2}} \times
        \nonumber \\
        \times
        \lrae{{\boldsymbol \Omega}_0^{\tt E}}{{\boldsymbol \Phi}({\boldsymbol f}^{(2)})\cdots
            \not\hskip-5pt{\boldsymbol \Phi}(\not\hskip-3pt{\boldsymbol f}^{(j)})
            \cdots{\boldsymbol \Phi}({\boldsymbol f}^{(n)})\,}_{{\boldsymbol {\mathcal E}}}\;.
    \end{multline}
Here, $\not\hskip-5pt{\boldsymbol \Phi}(\not\hskip-3pt{\boldsymbol f}^{(j)})$
indicates that one omits the term with index  $j^{\rm th}$
from the product.  Moreover,  the expression \eqref{n-Field Monomial}
is a multi-linear, symmetric function of the
${\boldsymbol f}^{(j)}$'s, so it is determined uniquely---using a
polarization identity---as a linear combination of powers
${\boldsymbol \Phi}(\boldsymbol{g})^n$, where $\boldsymbol{g}$ is one of the $2^n$
functions $\pm {\boldsymbol f}^{(1)} \pm \cdots \pm {\boldsymbol f}^{(n)}$.

Hence, the expectations
\[
\lrae{{\boldsymbol \Omega}_0^{\tt E}}{{\boldsymbol \Phi}({\boldsymbol f}^{(1)})\cdots{\boldsymbol \Phi}({\boldsymbol f}^{(n)})\,}_{{\boldsymbol {\mathcal E}}}\]
are determined uniquely by the characteristic function
$S({\boldsymbol f})$. From the recursion relation, we infer that
\[
        \lrae{{\boldsymbol \Omega}_0^{\tt E}}{{\boldsymbol \Phi}({\boldsymbol f})^n}_{{\boldsymbol {\mathcal E}}}
        = \begin{cases} 
            0\;,&\text{if $n$ is odd}\\
            (2k-1)!!\,\lra{\overline{{\boldsymbol f}},
        \boldsymbol{D}\,{\boldsymbol f}}_{{\boldsymbol{L}_2}}^k\;,&\text{if
        $n=2k$}\\
        \end{cases} \;,
\]
so $S({\boldsymbol f})=e^{-\tfrac{1}{2}\lra{\overline {\boldsymbol f}, \boldsymbol{D} {\boldsymbol f}}_{{\boldsymbol{L}_2}}}$.  Also $S({\boldsymbol f})=S(f_{+}, f_{-})$ as a consequence of the second identity in  \eqref{Two-Point-Charged-Matrix}. 
\end{proof}


\subsection{Charge Conjugation}
We define {\em charge conjugation} as a unitary transformation ${\bf U}_{c}$ on ${\boldsymbol {\mathcal E}}$.   We substitute $\Phi_{+}$ for the ordinary charged field $\varphi=\frac{1}{\sqrt2}\lrp{\varphi_{1}+i\varphi_{2}}$ representing positive charge  and~$\Phi_{-}$ for its hermitian conjugate $\varphi^{*}=\frac{1}{\sqrt2}\lrp{\varphi_{1}-i\varphi_{2}}$ representing negative charge. Thus, hermitian conjugation reverses the charge for the ordinary field, while charge conjugation of our classical field is determined by the action of the charge-conjugation matrix 
\begin{equation}
        \mathcal{C}
        =
        \bm{0}{1}{1}{0}=\mathcal{C}^*=\mathcal{C}^{\rm T}=\overline{\mathcal{C}}=\mathcal{C}^{-1}\;.
        \label{Charge_Conjugation}
\end{equation}
Note that $\mathcal{C}$ is a $2\times 2$ matrix that acts on the components of the field, \emph{i.e.}, 
	\[
		(\mathcal{C}{\boldsymbol \Phi})(x)
		= \begin{pmatrix}
		\Phi_{-}(x)\\
		\Phi_{+}(x) 
		\end{pmatrix}\;.
	\]
The unitary charge-conjugation operator ${\bf U}_{c} $ on
 ${\boldsymbol {\mathcal E}}$  is defined setting by 	
 	\[
		{\bf U}_{c} \,{\boldsymbol \Phi}({\boldsymbol f}^{(1)})\cdots{\boldsymbol \Phi}({\boldsymbol f}^{(n)})\, \boldsymbol\Omega_{0}^{\tt E} 
		=( \mathcal{C}{\boldsymbol \Phi})({\boldsymbol f}^{(1)})\cdots( \mathcal{C}{\boldsymbol \Phi})({\boldsymbol f}^{(n)})\, \boldsymbol\Omega_{0}^{\tt E} \;,
	\]
and ${\bf U}_{c}\, \boldsymbol\Omega_{0}^{\tt E} =\boldsymbol\Omega_{0}^{\tt E}$.

\subsection{Complex Conjugation}
Next we consider complex conjugation of the classical field, defined by
	\[
		\overline{\mathbf\Phi}(x)
		= \begin{pmatrix}
		\Phi_{+}^{*}(x)\\
		\Phi_{-}^{*}(x)
		\end{pmatrix}
		=\lrp{\mathbf\Phi^*}^{\rm T}(x)\;.
	\]
As the charge conjugation is connected with hermitian conjugation, it is natural that the existence of a positive measure on the fields entails that the expectation $\lrae{\Omega_{0}^{\tt E}}{\mathcal{C} {\boldsymbol \Phi}(x) \,{\boldsymbol \Phi}(x')\,}_{{\boldsymbol {\mathcal E}}}=\lrp{\mathcal{C}\boldsymbol{D}}(x,x')$  is the kernel of a positive transformation
\[
		\mathcal{C} \boldsymbol{D}
		= \bm{ D^{{\rm T}}}{0}{0}{ D} \text{on ${\boldsymbol{L}_2} \, $. }
\]
This will be the case, if and only if $D$ itself is positive on $L_{2}$. Thus, a positive measure is associated with the condition
\[
		0
		\leqslant 
		\sum_{j,j'=1}^{N} 
			c_{j} \overline{c_{j}} 
				S({\boldsymbol f}_{j}-\mathcal{C}\overline{{\boldsymbol f}_{j'}})\;,
\]
for any choice of ${\boldsymbol f}_{j}\in {\boldsymbol{L}_2}$ and $c_{j}\in\mathbb{C}$.

\subsection{Time Reflection}
Let $\boldsymbol{\Theta}$ denote the unitary time reflection on~${{\boldsymbol {\mathcal E}}}$, namely 
	\[
		\boldsymbol{\Theta} = \bm{\Theta}{0}{0}{\Theta}\;.
	\]
In order to recover the usual properties of the charged fields, we expect that they should transform under time inversion as 
\[
        		\boldsymbol{\Theta}\,{\boldsymbol \Phi}(x)\,\boldsymbol{\Theta}
        		= \mathcal{C} \,\overline{{\boldsymbol \Phi}}(\vartheta x)\;,
	\label{Time-Reflection_Chargedfield-1}
\]
or in components  
\begin{equation}
	\Theta \,\Phi_\pm(x)\, \Theta
		= \Phi_\mp^{*} (\vartheta x) \; . 
	\label{Time-Reflection_Chargedfield-1}
\end{equation}
In fact, we have the following criterion.

\begin{proposition}[\bf Time Reflection]\label{Prop:ChargedFieldTimeReflection}
The field transforms according to \eqref{Time-Reflection_Chargedfield-1} if and only if 
\begin{equation}
        \vartheta \sigma_\pm^{\rm T} \vartheta =
        \sigma_\mp^*\;.
    \label{Time-Reflection_Chargedfield-2}
\end{equation}
In this case, 
\begin{equation}
		(\vartheta\mathcal{C})\,\boldsymbol{D}\lrp{\vartheta\mathcal{C}}
        		=\boldsymbol{D}^*
        		= \overline {\boldsymbol{D}}\;,
    \label{Time-Reflection_Chargedfield-3}
\end{equation}
and the operator $\vartheta \mathcal{C}\boldsymbol{D}$ equals
\begin{equation}
        \vartheta\mathcal{C}\boldsymbol{D}
        =\bm{\vartheta D^{\rm T}}{0}{0}{\vartheta D}
        =\bm{\sigma_-^*\vartheta\sigma_-^{\phantom{*}}}{0}{0}{\sigma_+^*\vartheta
        \sigma_+^{\phantom *}}\;,
    \label{Time-Reflection_Chargedfield-4}
\end{equation}
which is self-adjoint on ${\boldsymbol{L}_2}$. 
\end{proposition}

\begin{proof}[\bf Proof] In terms of components,
\begin{eqnarray*}
        \Theta\,\Phi_\pm(x)\,\Theta
        &=&  \lrp{2\pi}^{-d/2}\int \Theta\,\widetilde Q _\pm(k)\,\Theta\,
                \widetilde \sigma_\pm(k)
                    \,e^{ik\cdot x}\,dk
        \nonumber \\
        &=&  \lrp{2\pi}^{-d/2}\int \widetilde Q _\pm(\vartheta k)
                \widetilde \sigma_\pm(k)
                    \,e^{ik\cdot x}\,dk
        \nonumber \\
        &=&  \lrp{2\pi}^{-d/2}\int \widetilde Q _\mp^{*} (-\vartheta k)
                \widetilde \sigma_\pm(k)
                    e^{ik\cdot x}\,dk
         \nonumber \\
         &=&  \lrp{2\pi}^{-d/2}\int \widetilde Q _\mp^{*} (k)
                \widetilde \sigma_\pm(-\vartheta k)
                    e^{-ik\cdot(\vartheta x)}\,dk
                    \nonumber \\
        &=&  \lrp{2\pi}^{-d/2}\lrp{\int \widetilde Q _\mp(k)
                \overline{\widetilde\sigma_\pm(-\vartheta k)}
                    \,e^{ik\cdot(\vartheta x)}\,dk}^*\;.
    \end{eqnarray*}
Thus, the desired equivalence
\eqref{Time-Reflection_Chargedfield-1} is equivalent to
$\widetilde \sigma_\pm$ satisfying the relation
$\overline{\widetilde \sigma_\pm(-\vartheta k)}=\widetilde
\sigma_\mp(k)$. In configuration space this is equivalent
to the operator relation \eqref{Time-Reflection_Chargedfield-2}.

Note that the self-adjoint, unitary operator on ${\boldsymbol{L}_2}$, given by 
\[
		\vartheta \mathcal{C}
		= \bm{0}{\vartheta}{\vartheta}{0}\;,
\]
yields
	\[
		 (\vartheta\mathcal{C})\,\boldsymbol{D}\lrp{\vartheta\mathcal{C}}
		 = \bm{0}{\vartheta D^{{\rm T}}\vartheta}{\vartheta D \vartheta}{0}\;.
	\]
Hence, relation  \eqref{Time-Reflection_Chargedfield-3} is equivalent to $\vartheta D\vartheta = D^{*}$ and $\vartheta D^{{\rm T}}\vartheta = D^{{\rm T} \,*}$, namely to the self-adjointness of $\vartheta D$ and of $\vartheta D^{{\rm T}}$  on $L_{2}$.  In the case that \eqref{Time-Reflection_Chargedfield-2} holds, one can use \eqref{Charged Covariance} to infer $ \vartheta D = \vartheta\,\sigma_{-}^{{\rm T}}\,\sigma_{+} = \sigma_+^*\,\vartheta\, \sigma_+$.  Thus, $\vartheta D$ is self-adjoint on $L_{2}$.  Likewise,
    $\vartheta D^{\rm T}
    =\vartheta\sigma_+^{\rm T}\sigma_-
    =\sigma_-^*\vartheta\sigma_-$
is self-adjoint on $L_{2}$.  As a consequence, one infers that 
\[
        \vartheta \mathcal{C} \boldsymbol{D}
        = \bm{\vartheta D^{\rm T}}{0}{0}{\vartheta D}
        =\bm{\sigma_-^*\vartheta\sigma_-^{\phantom{*}}}{0}{0}{\sigma_+^*\vartheta
        \sigma_+^{\phantom *}}\;,
\]
which has the desired form \eqref{Time-Reflection_Chargedfield-4} and is self-adjoint on ${\boldsymbol{L}_2}$.
\end{proof}

\subsection{Time-Reflection Positivity for the Charged Field}
To study time-reflection positivity of $\boldsymbol{\Theta}$ on ${{\boldsymbol {\mathcal E}}}_{+}$, take the Osterwalder-Schrader form for the charged field to be
\begin{equation}
        \lra{\ \cdot \ ,\ \cdot \ }_{{\boldsymbol {\mathcal H}}}
        = \lra{\ \cdot \ , \boldsymbol{\Theta}\ \cdot \
        }_{{\boldsymbol {\mathcal E}}}\;.
    \label{Charged OS Form}
\end{equation}
Let $ {\boldsymbol{L}_2}_{+}=L_{2}(\mathbb{R}^{d}_{+})\oplus L_{2}(\mathbb{R}^{d}_{+}) $.  
\begin{definition}\label{Defn:Charged RP}
The functional $S({\boldsymbol f})$ is time-reflection positive if 
\begin{equation}
		0
		\leqslant 
		\sum_{j,j'=1}^{N} 
			c_{j} \overline{c_{j'}} \; 
				S({\boldsymbol f}_{j}-\vartheta\mathcal{C}\overline{{\boldsymbol f}_{j'}})\;,
	\label{Charged TimeRP}
\end{equation}
for any choices of ${\boldsymbol f}_{j}\in {\boldsymbol{L}_2}_{+}$ and $c_{j}\in\mathbb{C}$, $j=1, \ldots, N$.  It is also time-reflection positive if  \eqref{Charged TimeRP} holds
for any choices of ${\boldsymbol f}_{j}\in {\boldsymbol{L}_2}_{-}$ and $c_{j}\in\mathbb{C}$, $j=1, \ldots, N$.  In case both conditions hold, then $S({\boldsymbol f})$ is doubly time-reflection positive.
\end{definition}

\begin{proposition}\label{Prop:RP Time Charged Field}
Time-reflection positivity on $\mathcal{E}_+$
is equivalent to the statement that
\begin{equation}
        0\leqslant \vartheta\mathcal{C}\boldsymbol{D}
        =\bm{\vartheta D^{\rm T}}{0}{0}{\vartheta D}
        =\bm{\sigma_-^*\vartheta\sigma_-^{\phantom{*}}}{0}{0}{\sigma_+^*\vartheta
        \sigma_+^{\phantom *}}\;
        \text{on} \;  {\boldsymbol{L}_{2,+}}\;.
    \label{Chraged-OS RP}
\end{equation}
\end{proposition}

\begin{proof}[\bf Proof]
As the characteristic function for the charged field is
Gaussian, we infer from the proof of Proposition~\ref{Prop:RP-Time
Equivalence} that positivity of the OS form~\eqref{Charged
OS Form} on ${{\boldsymbol {\mathcal E}}}$ is equivalent to positivity of the
two point function 
\[ 
\langle 	
	{\boldsymbol \Phi}({\boldsymbol f}){\boldsymbol \Omega}_0^{\tt E}\, ,
	{\boldsymbol \Phi}({\boldsymbol f}){\boldsymbol \Omega}_0^{\tt E}   
\rangle_{\boldsymbol {\mathcal H}}
\]
for
${\boldsymbol f}\in{\boldsymbol{L}_{2,+}}=L_2(\mathbb{R}^d_+)\oplus L_2(\mathbb{R}^d_+)$. For
${\boldsymbol f},\boldsymbol{g}\in{\boldsymbol{L}_{2,+}}$ we claim that the putative inner
product of two such vectors in ${\boldsymbol {\mathcal H}}$ is
\begin{equation}
    \label{ChargedReflectedTwoPointFunction}
        \lra{{\boldsymbol \Phi}({\boldsymbol f}){\boldsymbol \Omega}_0^{\tt E}\,,
        {\boldsymbol \Phi}(\boldsymbol{g}){\boldsymbol \Omega}_0^{\tt E}}_{\boldsymbol {\mathcal H}}
        = \lra{{\boldsymbol f},\vartheta \mathcal{C}\boldsymbol{D}\,
        \boldsymbol{g}}_{\boldsymbol{L}_2}\;,
\end{equation}
where 
\[
        \vartheta\mathcal{C}\boldsymbol{D}
        =\bm{\vartheta D^{\rm T}}{0}{0}{\vartheta D}
        =\bm{\sigma_-^*\vartheta\sigma_-}{0}{0}{\sigma_+^*\vartheta \sigma_+}\;.
\]
From \eqref{ChargedReflectedTwoPointFunction} we infer that
positivity of the OS form on ${{\boldsymbol {\mathcal E}}}$ is equivalent  to~\eqref{Chraged-OS RP}, as claimed.  In fact,
\begin{eqnarray*}
        \lra{{\boldsymbol \Phi}({\boldsymbol f}){\boldsymbol \Omega}_0^{\tt E}\,,
        {\boldsymbol \Phi}(\boldsymbol{g}){\boldsymbol \Omega}_0^{\tt E}}_{\boldsymbol {\mathcal H}}
        &=& \lra{{\boldsymbol \Phi}({\boldsymbol f}){\boldsymbol \Omega}_0^{\tt E}\,,
        \boldsymbol{\Theta}\,{\boldsymbol \Phi}(\boldsymbol{g}){\boldsymbol \Omega}_0^{\tt E}}_{{\boldsymbol {\mathcal E}}}
        \nonumber
        \\
        &=& \lra{{\boldsymbol \Omega}_0^{\tt E}\,,\boldsymbol{\Theta} \,{\boldsymbol \Phi}^*(\overline{\boldsymbol f})
        \boldsymbol{\Theta}\,{\boldsymbol \Phi}(\boldsymbol{g}){\boldsymbol \Omega}_0^{\tt E}}_{{\boldsymbol {\mathcal E}}}
                \nonumber
        \\
        &=& \lra{{\boldsymbol \Omega}_0^{\tt E}\,,{\boldsymbol \Phi}(\mathcal{C}\vartheta\overline{\boldsymbol f})
        \,{\boldsymbol \Phi}(\boldsymbol{g}){\boldsymbol \Omega}_0^{\tt E}}_{{\boldsymbol {\mathcal E}}}
        = \lra{{\boldsymbol f},\vartheta \mathcal{C}\boldsymbol{D}\,
        \boldsymbol{g}}_{\boldsymbol{L}_2}\;.
\end{eqnarray*}
In the final equality we use the fact that action of the matrix
$\mathcal{C}$ on ${\boldsymbol \Phi}$ is given by  
\eqref{Charge_Conjugation}.  Therefore
$(\mathcal{C}{\boldsymbol \Phi})(\vartheta{\boldsymbol f})={\boldsymbol \Phi}(\mathcal{C}\vartheta{\boldsymbol f})$,
and
\begin{equation}
        \lrae{{\boldsymbol \Omega}_0^{\tt E}}{(\mathcal{C}{\boldsymbol \Phi})(\vartheta{\boldsymbol f})\,
            {\boldsymbol \Phi}(\boldsymbol{g})\,}_{{\boldsymbol {\mathcal E}}}
        = \lra{ \overline{{\boldsymbol f}},\vartheta\mathcal{C}\boldsymbol{D}
        \boldsymbol{g}}_{\boldsymbol{L}_2}\;.
    \label{Gamma_Covariance}
\end{equation}
\end{proof}

\subsection{Spatial-Reflection Positivity for the Charged Field}
One can formulate spatial reflection positivity by
substituting $L_2(\mathbb{R}^d_{\vec n+})$ for $L_2(\mathbb{R}^d_+)$,
substituting  $\pi_{\vec n}$ for~$\vartheta $, substituting
$\Pi_{\vec n}$ for $\Theta$, and substituting ${{\boldsymbol {\mathcal E}}}_{\vec
n+}$ for ${{\boldsymbol {\mathcal E}}}_{+}$ in \eqref{Charged OS Form}.  Spatial
reflection positivity gives a pre-inner product
\[
        \lra{A, B}_{\boldsymbol {\mathcal H} (\vec n)}
        = \lra{A,\Pi_{\vec n} B}_{{\boldsymbol {\mathcal E}}}\;,
 \]
on ${{\boldsymbol {\mathcal E}}}_{\vec n+, 0}$, as well as  new requirements on
the transformation of the field and of $D$.   

\begin{proposition}[\bf Spatial-Reflection]
\label{Prop:RP Charged-Spatial} If one replaces $\vartheta, \Theta,
\boldsymbol{\Theta}, \mathbb{R}^d_+, {{\boldsymbol {\mathcal E}}}_{+}$ in the statement of
Proposition~\ref{Prop:ChargedFieldTimeReflection}  by $\pi_{\vec n},
\Pi_{\vec n}, \mathbf{\Pi}_{\vec n}, \mathbb{R}^d_{\vec n+}, {{\boldsymbol {\mathcal E}}}_{\vec n+}\,$,
respectively, then the proposition remains valid.
\end{proposition}

\begin{proposition}[\bf Spatial-Reflection
Positivity]\label{Prop:RP Charged-Spatial} If one replaces
$\vartheta, \Theta, \boldsymbol{\Theta}, \mathbb{R}^d_+, {{\boldsymbol {\mathcal E}}}_{+}$ in the statement of
Proposition~\ref{Prop:RP Time Charged Field}  by $\pi_{\vec n},
\Pi_{\vec n}$, $\mathbf{\Pi}_{\vec n}$, $\mathbb{R}^d_{\vec n+}$, ${{\boldsymbol {\mathcal E}}}_{\vec n+}\,$,
respectively, then the proposition remains valid.
\end{proposition}

\setcounter{equation}{0}
\section{Compactification
\label{Sect:Compactification_RP}} In this section we show that reflection positivity 
carries over when we compactify one coordinate $x_{j}\in \mathbb{R}$ on the line to a corresponding coordinate on a circle $x_{j} \in S^{1}$.

We  consider spacetimes of the general form ${\bf X}=X_1\times \cdots
\times X_d$, where each factor $X_i$ either equals $\mathbb{R}$
(the real line) or $S^{1}$ (a circle of length $\ell_j$).
Denote the first coordinate by $t=x_0$ and let $x=(t,\vec
x)$.  {\em Compactification} of a coordinate $x_j\in X_j$
means replacing the coordinate $x_j\in\mathbb{R}$ by  a
corresponding coordinate $x_j\in S^{1}$.  One
denotes this compactification as
\begin{multline}
        \mathbf{X}
        = X_1\times \cdots \times X_{j-1}\times
                \mathbb{R}\times X_{j+1}\times\cdots \times X_d
        \; \longrightarrow \;
        \nonumber \\
        \; \longrightarrow \;
\mathbf{X}^{cj}
        = X_1\times \cdots \times X_{j-1}\times
            S^{1} \times X_{j+1}\times\cdots \times X_d\;.
\end{multline}
Of course one can take $\mathbf{X}^{cj}$ as a new $\mathbf{X}$, and
continue to compactify.  The minimally compactified space
is $\mathbb{R}^d$; the maximally compactified space is the torus
$\mathbb{T}^d$, with periods $\beta=\ell_0,\ell_1,\cdots,
\ell_{d-1}$.

Parameterize the compactified coordinate $x_j$ as
    \[
        x_j
        \in S^{1}
        = [-\tfrac12\ell_j\,,\,\tfrac12\ell_j]\;,
    \]
and the positive-subspace of $X_j$ by
\[
        X_{j+}
        = (S^{1}_{+})_j = [0\,,\,\tfrac12\ell_j]\;,
        \quad \text{or}\quad
        X_{j+}=\mathbb{R}_+ = [0,\infty)\;.
\]
Likewise, denote the $j$-positive subspace of $\mathbf X$ by
\begin{equation}
        {\mathbf X}_{j+}=X_1\times\cdots \times X_{j+}\times \cdots\times X_d\;.
    \label{j-PositiveSubspace}
\end{equation}
The reflection of the coordinate $x_j$ is the transformation
\[
        \pi_j \colon x\to \pi_j\, x\;,
        \text{where } (\pi_j\,x)_i=
            \begin{cases}
                \phantom{-} x_i \;,&\text{if} \quad  i\neq j\\
                -x_j\;,&\text{if} \quad i=j
            \end{cases} 
        \;.
\]
We denote $\vartheta $ by $\pi_0$, so we treat time reflection and
spatial reflection on an equal footing. The Fock space $\mathcal{E}$
over $\mathcal{K}$ and the $j$-positive subspaces $\mathcal{E}_{j+}$ will now refer to 
the one-particle space $\mathcal{K}=L_2({\bf X})$, and the
$j$-positive subspace  $\mathcal{K}_{j+}= L_2({\mathbf
X}_{j+})$.

\begin{definition}
The operator $D$ is  doubly-reflection-positive with respect to $\pi_{j}$  if both 
\begin{equation}
		0\leqslant  \pi_{j} D\;,
		\quad \text{and} \quad 
		0\leqslant  D\pi_{j}
		\quad \text{on} \quad
		L_2(\mathbf{X}_{j+})\;.
	\label{Doubly-Reflection-Positive}
\end{equation}
This is  equivalent to both
\begin{equation}
		0\leqslant  \pi_{j} D\;,
		\quad \text{and} \quad
		0\leqslant  D\pi_{j}
		\quad \text{on} \quad
		L_2(\mathbf{X}_{j-})\;,
	\label{Doubly-Reflection-Positive-2}
\end{equation}
or to 
\begin{equation}
		0\leqslant  \pi_{j} D \quad \text{on both} \quad L_2(\mathbf{X}_{j\pm})\;.
	\label{Doubly-Reflection-Positive-3}
\end{equation}
\end{definition}

Recall that symmetry of the operator $D$ with integral kernel $D(x,x')$ means  $D(x,x')=D(x',x)$.  Moreover, a covariance for classical fields that is reflection-positive with respect to the reflection~$\pi$  must be symmetric and satisfy $\pi D \pi = D^{*}$.

\begin{proposition}
Suppose  $D$ is  symmetric on $L_{2}(\mathbf{X})$, and 
\[ 
\pi D\pi= D^{*}
\]
 for a reflection $\pi$.  Then double-reflection-positivity and reflection-pos\-itivity with respect to $\pi$ on $L_{2}(\mathbf{X}_{+})$  are equivalent. 
\end{proposition}

\begin{proof}[\bf Proof]  Double RP is a stronger condition, so one only needs to show that  RP ensures double RP.  Since $\pi X_{\pm}=X_{\mp}$ and $\pi$ is self-adjoint and unitary, the condition \eqref{Doubly-Reflection-Positive-3} is equivalent to the other two conditions.  Thus, it is sufficient to show that $0\leqslant  \pi D$ on $L_{2}(X_{+})$ implies $0\leqslant  D\pi$ on the same subspace.   

Using $\pi D\pi =D^{*}$, one infers 
\[
D(x,x')=(\pi D^{*} \pi)(x,x')= ( D^{*} )(\pi x,\pi x')=  \overline{ D(\pi x',\pi x)}   \; , 
\]
so symmetry ensures that $D(x,x')=\overline{{D(\pi x, \pi x')}}$.  This is the operator relation $D= \overline{\pi D \pi} = \pi \overline{ D} \pi$, as $\pi $ is real, so  also $\overline{D \pi} =\pi D$. Therefore  
\begin{equation}
		 \lra{{f}, D\pi  {f}}_{L_{2}(\mathbf{X})} 
		 =\overline{ \lra{\overline {f},\overline{ D\pi}\,\overline {f}}}_{L_{2}(\mathbf{X})}
		 =\overline{ \lra{\overline {f},{ \pi D}\,\overline {f}}}_{L_{2}(\mathbf{X})}\;.
\end{equation}
Complex conjugation leaves $L_{2}(\mathbf{X}_{+})$ invariant, so reflection-positivity ensures $0\leqslant  \lra{\overline {f},{ \pi D}\,\overline {f}}_{L_{2}(\mathbf{X})}$.  Therefore $0 \leqslant  \lra{f, D \pi f}_{L_{2}(\mathbf{X})}$ for  $f\in L_{2}(\mathbf{X}_{+})$, as claimed.
\end{proof}

Let $e_j$ denote a unit vector in the $j^{\rm
th}$ coordinate direction. Define~$T_j$ as the unitary translation operator  on spacetime that translates by one  period  $\ell_{j}$ in the  coordinate direction $j$, namely 
\[
		(T_{j}f)(x)
		= f(x-\ell_{j} e_{j})\;,
\]
with $T_{j}:L_{2}(\mathbf{X}_{j+})\to L_{2}(\mathbf{X}_{j+}) $.  
In addition, 
\[
		\lrp{\pi_{j}  T_{j}^{-1/2}f}(x)
		=\lrp{T_{j}^{-1/2}f}(\pi_{j}x)
		=f(\pi_{j}x +\tfrac12\ell_{j}e_{j})\;.
\]
Now assume that $D(x-x')$
decreases sufficiently rapidly so that the sum
\begin{equation}
        D^{cj}(x,x')
        =  \sum_{n=-\infty}^\infty \lrp{T_{j}^{-n}D}(x-x')
        = \sum_{n=-\infty}^\infty D(x-x'+n\ell_j \,e_j)\;
    \label{Kernel-Periodization}
\end{equation}
converges absolutely.   Then \eqref{Kernel-Periodization} defines the  operator $D^{cj}$ on the space~$\mathbf{X}^{cj}$.  

\begin{proposition}[\bf Compactification of Reflection Positivity]
\label{Prop:Compactification RP}
Let $D$ be translation-invariant and symmetric on $L_2(\mathbf{X})$ and let $D$ be doubly-reflection-positive with respect to $\pi_{j}$.  Define $D^{ci}$ by the integral kernel~\eqref{Kernel-Periodization}. Then $D^{ci}$  is doubly-reflection-positive with respect to $\pi_{j}$.
\end{proposition}

\begin{remark}
If $i=j$, the proposition states that reflection positivity for a coordinate $x_j\in\mathbb{R}$ extends to the case of a compactified coordinate $x_j\in S^{1}$.  However, for $i\neq j$ the proposition says that reflection positivity in the $j^{\rm th}$-direction remains unaffected by the compactification of spacetime along a different coordinate direction $x_i$.
\end{remark}

\begin{proof}[\bf Proof] {\bf The Case} $i=j$.
We first show that $0\leqslant  \lra{f,\pi_j  D^{cj} \,f}_{L_2(\mathbf{X}^{cj})}$  for all
$f\in L_2(\mathbf{X}^{cj}_{j+})$. Here, $f$ depends on the coordinate $x_{j}\in[-\tfrac{1}{2}\ell_{j}, \tfrac{1}{2}\ell_{j}]$ and is supported in the positive half interval.  Imbed $L_2(\mathbf{X}^{cj})$ in $L_2(\mathbf{X})$ in the natural way, so that translations $T_{j}$ on $L_{2}(\mathbf{X})$ translate the support of $f$ by $\ell_{j}e_{j}$.  Then
\begin{eqnarray}
		\lra{f,\pi_j  D^{cj} f}_{L_2(\mathbf{X}^{cj})}
		&=&\lra{f,\pi_j  D^{cj} f}_{L_2(\mathbf{X})}
		\nonumber \\
		&=&  \sum_{n=-\infty}^{\infty}\lra{f,\pi_j  T_{j}^{-n} D f}_{L_2(\mathbf{X})} \;.
    \label{Translation-Sum}
\end{eqnarray}
For the moment, let us suppress the variables $x_{i}$ for $i\neq j$.  Thus, we write
\begin{eqnarray}
		 \lra{f, \pi_j  T_j^{-n} D f}_{L_2(\mathbf{X})} \kern -2mm
	         &=& \kern -2mm \int_{0}^{\tfrac{1}{2}\ell_{j}} \kern -2mm \int_{0}^{\tfrac{1}{2}\ell_{j}} 
         	\overline{f(x_{j})}  D(-x_{j}-x'_{j}+n\ell_{j}) f(x'_{j}) dx_{j} dx'_{j}\nonumber \\
		&=&  \lra{T_j^{-n/2}f, \pi_j  D\,T_j^{-n/2} f}_{L_2(\mathbf{X})}
		\;.
	\label{One-Positive-Term}
 \end{eqnarray}

We now show that each term in the sum \eqref{Translation-Sum} is positive. 
These terms are of the form given in~\eqref{One-Positive-Term}.  In case $n\leqslant  0$, the operator 
$T_{j}^{-n/2}$ maps $L_{2}(\mathbf{X}_{j+})$ into itself. Thus, reflection positivity of~$\pi_{j}D$ ensures positivity, \emph{i.e.},
\begin{equation}
		0\leqslant   \lra{f, \pi_j  T_j^{-n} D f}_{L_2(\mathbf{X})}
		\quad \text{for} \quad n\leqslant  0\;.
	\label{Positive n Terms} 
\end{equation}

On the other hand, if $n\geqslant 1$, then $T_{j}^{-n/2}$ maps $L_{2}(\mathbf{X}_{j+})$  into  $L_{2}(\mathbf{X}_{j-})$.  
Hence, $\pi_{j} T_{j}^{-n/2}$ 
maps~$L_{2}(\mathbf{X}_{j+})$  into itself.  Thus, for $f\in L_{2}(\mathbf{X}_{j+})$,
\begin{eqnarray*}
		 \lra{f, \pi_j  T_j^{-n} D f}_{L_2(\mathbf{X})}
		 &= &\lra{ T_j^{n/2}\pi_{j }f,  D  T_j^{-n/2}f}_{L_2(\mathbf{X})}
		 \nonumber \\
		 &=  &\lra{\pi_{j } T_j^{-n/2} f, ( D\pi_{j}  )\pi_{j}  T_j^{-n/2}f}_{L_2(\mathbf{X})}\;.
\end{eqnarray*}
The positivity of $D\pi_{j}$ on $L_{2}(\mathbf{X}_{j+})$ now ensures that  
\begin{equation}
		0\leqslant   \lra{f, \pi_j  T_j^{-n} D f}_{L_2(\mathbf{X})}\; 
		\quad \text{for} \quad 1\leqslant   n\;.
	\label{Other Terms}
\end{equation}
The relations \eqref{Positive n Terms} and \eqref{Other Terms} show that  \eqref{Translation-Sum} is a sum of non-negative terms.  Thus, $0\leqslant \pi_{j} D^{cj}$ on $ L_2(\mathbf{X}^{cj}_{j+})$, as claimed.  

One can reduce the proof of  positivity of $D^{cj}\pi_{j} $  on $ L_2(\mathbf{X}^{cj}_{j+})$ to the previous case.  In fact, one can replace $n$ by $-n$ in \eqref{Kernel-Periodization} and write 
\begin{eqnarray*}
         \lra{f, D^{cj} \pi_j f}_{L_2(\mathbf{X})}
          &=&  \sum_{n=-\infty }^{\infty} \lra{f,  T_j^{n} D \pi_j f}_{L_2(\mathbf{X})}\nonumber \\
          &=&  \sum_{n=-\infty }^{0} \lra{T_j^{-n/2} f,   (D \pi_j ) \,T_j^{-n/2}f}_{L_2(\mathbf{X})}
          \nonumber \\
          && \quad +   \sum_{n=1 }^{\infty} \lra{\pi_{j}T_j^{-n/2}f, (\pi_{j}D)\, \pi_{j} T_j^{-n/2} f}_{L_2(\mathbf{X})}\;.
\end{eqnarray*}
We have already shown that these matrix elements  are positive. Thus, $0\leqslant  D^{cj}\pi_{j} $  on $ L_2(\mathbf{X}^{cj}_{j+})$.

\bigskip
\noindent
{\bf The Case $i\neq j$.}
For $f\in L_2(\mathbf{X}^{ci}_{j+})$ one has, for any $n'\in \mathbb{Z}$, 
\begin{eqnarray*}
         \lra{f,\pi_j  D^{ci} f}_{L_2(\mathbf{X}^{cj})}
         &=& \sum_{n=-\infty }^\infty
         \lra{f, \pi_j  T_i^{-n} D f}_{L_2(\mathbf{X})}
         \nonumber \\
         &=&  \sum_{n=-\infty }^\infty
         \lra{f, \pi_j  T_i^{-n+n'} D f}_{L_2(\mathbf{X})}\;.
\end{eqnarray*}
Using the fact that $T_{i} $ commutes with both $\pi_{j}$ and $D$, one arrives at 
\begin{eqnarray}
         \lra{f,\pi_j  D^{ci} f}_{L_2(\mathbf{X}^{cj})}
         &=&   \frac{1}{2N+1}  \sum_{n=-\infty }^\infty \sum_{n'=-N}^{N}
         \lra{T_{i}^{-n'}f, \pi_j   D  T_{i}^{-n}f}_{L_2(\mathbf{X})}  \nonumber \\
       &=&  \lim_{N\to\infty}
                \lra{g_N,  \pi_j  D\,g_N}_{L_2(\mathbf{X})}\;,
    \label{gN-limit}
\end{eqnarray}
where
\begin{equation}
        g_N
        = \frac{1}{\sqrt{2N+1}}\,
            \sum_{n=-N}^N T_i^{-n}  f \in  L_2(\mathbf{X}_{j+}) \;.
    \label{gN-RP}
\end{equation}
Observe that the translations $T_i$ act in the $i^{\rm th}$-coordinate
direction, orthogonal to the $j^{\rm th}$-coordinate
direction, so they map $L_2(\mathbf{X}_{j+})$ into itself.   As $\pi_j\,D$ is positive on
$L_2(\mathbf{X}_{j+})$, one concludes that
\[
		0\leqslant \lra{g_N, \pi_j D\,g_N}_{L_2(\mathbf{X})}
		\to \lra{f,\pi_j  D^{ci} f}_{L_2(\mathbf{X}^{cj})}\;.
\]
Thus, $0\leqslant  \pi_{j} D^{ci}$ on  $L_2(\mathbf{X}_{j+})$, as claimed.     The argument to show that $0\leqslant  D^{ci} \,\pi_{j}$ on $L_2(\mathbf{X}_{j+})$ is similar, so we omit the details.
\end{proof}

\setcounter{equation}{0}
\section{Summary of Positivity Conditions}
We now summarize the various positivity conditions that we have discussed in this paper in the case $X=\mathbb{R}^{d}$.  We state the conditions 
the characteristic function $S(f)$ of the field is subject to, and---for the Gaussian case---the conditions the covariance of the characteristic function
has to obey.   
\subsection{Neutral Fields}  
We use the unitary time-reflection operator $\vartheta$,  the unitary reflection $\pi_{\vec n}$ in the plane orthogonal to~$\vec n$, and in the Gaussian case the covariance $D$.  

\paragraph{\em Measure Positivity:}  The condition 
\begin{equation}
		0\leqslant  \sum_{j,j'=1}^{N} \overline{c_{j}} c_{j'}\,  S(f_{j'}-\overline{f_{j}})\;,
	\label{NeutralMeasurePositivity}
\end{equation}
leads to the existence of a positive measure as the Fourier transform of $S(f)$.  In the Gaussian case  $S(f)=e^{-\tfrac{1}{2} \lra{\overline f, D f}_{L_{2}}}$, and condition \eqref{NeutralMeasurePositivity} is equivalent to 
\begin{equation}
		0\leqslant  D  \quad \text{on} \quad L_{2}(\mathbb{R}^{d})\;.
	\label{NeutralMP}
\end{equation}

\paragraph{\em Time-Reflection Positivity:}
The condition of time-reflection positivity is 
\begin{equation}
		0\leqslant  \sum_{j,j'=1}^{N} \overline{c_{j}} c_{j'}\,  S(f_{j'}-\vartheta \overline{ f_{j}})
		\quad \text{for} \quad
		f_{j}\in L_{2}(\mathbb{R}^{d}_{+})\;.
	\label{NeutralTRP}
\end{equation}
In the Gaussian case this is equivalent to 
\begin{equation}
		0\leqslant  \vartheta D
		\quad \text{on} \quad L_{2}(\mathbb{R}^{d}_{+})\;.
	\label{GNeutralTRP}	
\end{equation}

\paragraph{\em Alternative Time-Reflection Positivity:}
The alternative condition of time-reflection positivity is  
\begin{equation}
		0\leqslant  \sum_{j,j'=1}^{N} \overline{c_{j}} c_{j'}\,  S(f_{j'}-\vartheta \overline{ f_{j}})
		\quad \text{for} \quad
		f_{j}\in L_{2}(\mathbb{R}^{d}_{-})\;.
	\label{ANeutralTRP}
\end{equation}
In the Gaussian case this is equivalent to 
\begin{equation}
		0\leqslant  D\,\vartheta
		\quad \text{on} \quad
		L_{2}(\mathbb{R}^{d}_{+})\;.
	\label{AGNeutralTRP}	
\end{equation}

\paragraph{\em Spatial-Reflection Positivity:}
The condition for spatial-reflection positivity is  
\begin{equation}
		0\leqslant  \sum_{j,j'=1}^{N} \overline{c_{j}} c_{j'}\,  S(f_{j'}-\pi_{\vec n} \overline{ f_{j}})
		\quad \text{for} \quad
		f_{j}\in L_{2}(\mathbb{R}^{d}_{\vec n+})\;.
	\label{NeutralSRP}
\end{equation}
In the Gaussian case this is equivalent to 
\begin{equation}
		0\leqslant  \pi_{\vec n} \,D
		\quad \text{on} \quad L_{2}(\mathbb{R}^{d}_{\vec n+})\;.
	\label{GNeutralSRP}	
\end{equation}

\paragraph{\em Alternative Spatial-Reflection Positivity:}
The alternative condition for spatial-reflection positivity is  
\begin{equation}
		0\leqslant  \sum_{j,j'=1}^{N} \overline{c_{j}} c_{j'}\,  S(f_{j'}-\pi_{\vec n} \overline{ f_{j}})
		\quad \text{for} \quad 
		f_{j}\in L_{2}(\mathbb{R}^{d}_{\vec n-})\;.
	\label{ASRP}
\end{equation}
In the Gaussian case this is equivalent to 
\begin{equation}
		0\leqslant  D\,\pi_{\vec n}
		\quad \text{on} \quad L_{2}(\mathbb{R}^{d}_{\vec n+})\;.
	\label{AGNeutralSRP}	
\end{equation}

\subsection{Charged Fields}  
In the case of the charged field we use the  charge conjugation operator $\mathcal{C}$, the unitary time-reflection operator~$\vartheta$, and the unitary reflection $\pi_{\vec n}$ in the plane orthogonal to $\vec n$.  In the Gaussian case we also use the matrix covariance $\boldsymbol{D}$.

\paragraph{\em Measure Positivity:}  The positivity condition 
\begin{equation}
		0\leqslant  \sum_{j,j'=1}^{N} \overline{c_{j}} c_{j'}\,  S({\boldsymbol f}_{j'}-\mathcal{C}\overline{{\boldsymbol f}_{j}})\;,
	\label{ChargedMeasurePositivity}
\end{equation}
leads to the existence of a positive measure as the Fourier transform of $S({\boldsymbol f})$.  In the Gaussian case this positivity is equivalent to 
\begin{equation}
		0\leqslant  \mathcal{C}\boldsymbol{D}  \quad \text{on} \quad {\boldsymbol{L}_2}\;.
	\label{ChargedMP}
\end{equation}

\paragraph{\em Time-Reflection Positivity:}
The condition of time-reflection positivity is  
\begin{equation}
		0\leqslant  \sum_{j,j'=1}^{N} \overline{c_{j}} c_{j'}\,  S({\boldsymbol f}_{j'}-\vartheta\mathcal{C} \overline{{\boldsymbol f}_{j}})
		\text{for}
		{\boldsymbol f}_{j}\in {\boldsymbol{L}_2}_{+}\;.
	\label{ChargedTRP}
\end{equation}
In the Gaussian case this is equivalent to 
\begin{equation}
		0\leqslant  \vartheta \mathcal{C} \boldsymbol{D}
		\quad \text{on} \quad  {\boldsymbol{L}_2}_{+}\;.
	\label{GChargedTRP}	
\end{equation}

\paragraph{\em Alternative Time-Reflection Positivity:}
The alternative condition of time-reflection positivity is 
\begin{equation}
		0\leqslant  \sum_{j,j'=1}^{N} \overline{c_{j}} c_{j'}\,  S({\boldsymbol f}_{j'}-\mathcal{C}\vartheta \overline{ {\boldsymbol f}_{j}})
		\quad \text{for} \quad
		f_{j}\in {\boldsymbol{L}_2}_{-}\;.
	\label{AChargedTRP}
\end{equation}
In the Gaussian case the alternative condition is equivalent to 
\begin{equation}
		0\leqslant  \boldsymbol{D}\vartheta\mathcal{C}
		\quad \text{on} \quad  {\boldsymbol{L}_2}_{+}\;.
	\label{AGChargedTRP}	
\end{equation}

\paragraph{\em Spatial-Reflection Positivity:}
The condition for spatial-reflection positivity is 
\begin{equation}
		0\leqslant  \sum_{j,j'=1}^{N} \overline{c_{j}} c_{j'}\,  S({\boldsymbol f}_{j'}-\pi_{\vec n}\mathcal{C} \overline{ {\boldsymbol f}_{j}})
		\quad \text{for} \quad 
		{\boldsymbol f}_{j}\in {\boldsymbol{L}_2}_{,\vec n+}\;.
	\label{ChargedSRP}
\end{equation}
Here, ${\boldsymbol{L}_2}_{,\vec n+}= L_{2}(\mathbb{R}^{d}_{\vec n+})\oplus L_{2}(\mathbb{R}^{d}_{\vec n+})$ and  ${\boldsymbol{L}_2}_{,\vec n-}=\pi_{\vec n}{\boldsymbol{L}_2}_{,\vec n+}$.  In the Gaussian case \eqref {ChargedSRP} is equivalent to 
\begin{equation}
		0\leqslant  \pi_{\vec n} \,\mathcal{C}\boldsymbol{D}
		\quad \text{on} \quad {\boldsymbol{L}_2}_{,\vec n+}\;.
	\label{GChargedSRP}	
\end{equation}

\paragraph{\em Alternative Spatial-Reflection Positivity:}
The alternative condition for spatial-reflection positivity is  
\begin{equation}
		0\leqslant  \sum_{j,j'=1}^{N} \overline{c_{j}} c_{j'}\,  S({\boldsymbol f}_{j'}-\pi_{\vec n}\mathcal{C} \overline{ {\boldsymbol f}_{j}})
		\quad \text{for} \quad 
		{\boldsymbol f}_{j}\in {\boldsymbol{L}_2}_{,\vec n-}\;.
	\label{AChargedSRP}
\end{equation}
In the Gaussian case this is equivalent to 
\begin{equation}
		0\leqslant  \mathcal{C}\boldsymbol{D}\,\pi_{\vec n} 
		\quad \text{on} \quad  {\boldsymbol{L}_2}_{,\vec n+}\;.
	\label{AGChargedSRP}	
\end{equation}

\end{document}